\documentclass[11pt]{article}
\usepackage{amsfonts}
\usepackage{amsmath}
\usepackage{amsthm}
\usepackage{dsfont}
\usepackage{times}
\usepackage{epsfig}
\usepackage{framed}
\usepackage{verbatim}
\usepackage{sidecap}
\usepackage[small]{caption}
\usepackage{bibtopic}
\usepackage{color}

\newcommand{\ind}{{\mathds{1}}}
\newcommand{\bs}{\boldsymbol}
\DeclareMathOperator{\Prob}{Pr}
\DeclareMathOperator{\Ex}{\mathds{E}}
\DeclareMathOperator*{\argmax}{\arg\max}

\newtheorem{theorem}{Theorem}[section]
\newtheorem{lemma}[theorem]{Lemma}
\newtheorem{proposition}[theorem]{Proposition}

\newenvironment{example}[1][Example]{\begin{trivlist}
\item[\hskip \labelsep {\bfseries #1}]}{\end{trivlist}}

\title{  {\em \Large Conditional Modeling and the Jitter Method of Spike Re-sampling: Supplement} }

\author{Asohan Amarasingham \and Matthew T.~Harrison \and Nicholas G.~Hatsopoulos \and Stuart Geman}
\date{}

\DeclareCaptionLabelFormat{sfig}{\textbf{Figure S#2}}
\begin{document}

\maketitle

\begin{leftbar}
This technical report accompanies the manuscript ``Conditional Modeling and the Jitter Method of Spike Re-sampling,'' \cite{jitter} and provides further details, comments, references, and equations that were omitted from this main text in the interest of brevity. To ease referencing, the sectioning of the report follows that of the main text.

A few of our remarks in the supplement may be of interest to a broad audience.  For quick identification, these high-level remarks are bordered by a left vertical bar, like the one to the left of this paragraph.  The bulk of this document, however, contains technical details about the various simulations and data analyses presented in the main text.  These would primarily be of interest to a reader who was hoping to reproduce our methods exactly.  There is also a self-contained Mathematical Appendix at the end of the supplement that provides a more formal treatment of jitter.     
\end{leftbar}

\section{Introduction}

\subsubsection*{Note on terminology}

\begin{leftbar}
Some authors prefer alternative terms for jittering, such as ``dithering'' \cite{gerstein,grun09,pazienti07,pazienti08}, ``teetering'' \cite{shmiel05}, or ``artificial jitter'' \cite{rokem}, etc., presumably to distinguish from another use of the word ``jitter'' as the intrinsic temporal variability in individual spikes. For example, in some situations involving highly reliable \cite{billimoria} or simulated or modeled \cite{pazienti07} spike trains, individual spikes can unambiguously be placed in correspondence with one another across trials. In that case, the temporal variability in a spike's timing, under this correspondence, can be quantified directly and is commonly called the ``spike jitter'' \cite{billimoria,tiesinga}. In this paper, we continue to use ``jitter'' in its resampling sense, leaving these two uses of the term to be disambiguated by context.
\end{leftbar}

\subsubsection*{Hypothesis testing and p-values}

The introduction mentions hypothesis testing.  Recall that a hypothesis test consists of
\begin{itemize}
\item A null hypothesis, often denoted $H_0$, which is a collection of hypothetical distributions for the data.
\item An alternative hypothesis, often denoted $H_a$ or $H_1$, which is another collection of distributions for the data, disjoint from $H_0$.
\item A critical region or rejection region, say $C$, which is a collection of possible outcomes of the data for which we reject $H_0$ whenever the data is in $C$.
\end{itemize}
We evaluate the performance of a hypothesis test by its error probabilities.  For a probability distribution $P\in H_0$, an error (a type I error) is made when the data occurs in the critical region. The probability of a type I error is $P(C)$.  For $P\not\in H_0$, including $P\in H_a$, an error (a type II error) is made when the data does not occur in the critical region. The probability of a type II error is $1-P(C)$.  The size of a hypothesis test is the maximum error probability for all distributions in $H_0$, i.e., $\max_{P\in H_0} P(C)$.  When designing a hypothesis test, one usually tries to keep the size below some prespecified level, say $0.05$.  The specification of $H_a$ is useful for choosing among tests whose size does not exceed this level, with the idea being to choose a test that has good power, i.e., low error probabilities, for distributions in $H_a$.  Consequently, $H_a$ is important for interpreting a rejection of $H_0$.  Note, however, that only $H_0$ is used for quantifying the size of a test.  Note also that there may be distributions in neither $H_0$ nor $H_a$, in which case a rejection of $H_0$ is considered the correct decision, even though $H_a$ does not include the data distribution.  This is an example of model misspecification and it can lead to erroneous scientific conclusions.

It is common to create a critical region via a test statistic, $T$, which is just a scalar summary of the data, and a critical value, $a$, so that the critical region is the event $C=\{T \geq a\}=\{x:T(x)\geq a\}$ (or some similar set, like $\{T \leq a\}$), where $x$ denotes a dataset.  A p-value, say $p$, is a special type of test statistic with the property that a critical region of the form $C=\{p\leq a\}$ always creates a hypothesis test of level $a$, that is, of size $\leq a$, for any $0\leq a\leq 1$.  In other words, $P(p\leq a)\leq a$ for all $P\in H_0$.  Given an ordinary test statistic $T$ and critical regions of the form $C=\{T\geq a\}$, one can always create a p-value via $p(x) = \max_{P\in H_0} P\bigl(\{x':T(x') \geq T(x)\}\bigr)$, where $x$ and $x'$ denote datasets.  P-values are useful for communicating hypothesis tests since each reader can choose his or her desired threshold for maximum probability of type I error.  We refer the reader to \cite{lehmann2005testing}, for example, for many more details and examples about hypothesis testing.

\subsubsection*{Poisson processes}

\begin{leftbar}
Poisson processes and their generalizations are the most basic types of point processes.  Although jitter methods are applicable for many non-Poisson processes (which is fortunate, since neural spike trains are often poorly approximated by Poisson processes \cite{amarasingham-chen-geman-harrison-sheinberg-poisson,kass-ventura-brown}), we use Poisson processes in many examples because of their simplicity and familiarity.
\end{leftbar}

Imagine partitioning time into extremely fine bins of length $\Delta$.  In each bin we (independently) flip a coin with probability $\lambda\Delta$ for heads.  Heads means that we observe an event in that time bin.  Tails means that we do not.  A (homogeneous) Poisson process is the generalization of this procedure for infinitesimally small time bins (that is, holding $\lambda$ fixed and letting $\Delta\to 0$).  If the value of $\lambda$ is allowed to vary across time bins, so that we get a function $\lambda(t)$ for infinitesimally small bins, then this is an inhomogeneous Poisson process with intensity (or rate) function $\lambda(t)$.  If we first randomly choose the function $\lambda(t)$ itself from some collection of possible functions, then this is a Cox process.  We refer the reader to \cite{daley2003introduction}, for example, for details about these and other point process models.

\subsubsection*{Trial-to-trial variability}

\begin{leftbar}
The introduction also mentions ``trial-to-trial variability'' and suggests that trial-to-trial variability may be as much a result of model misspecification as it is a result of non-stationarity in the data.  We would like to illustrate this comment with a simple example called ``amplitude variability'' or ``excitability variability'' \cite{brody99b}.  

Consider repeated independent trials of a point process.  Suppose we generate the data on each trial by first choosing randomly and independently a firing rate, say 50\% of the time the firing rate is 10 Hz and 50\% of the time the firing rate is 20 Hz.  Then we generate the spike train as a homogeneous Poisson process with the selected firing rate.  On the next trial we randomly choose the rate again from one of the two choices, and then generate the spikes.  The resulting process is a Cox process, since each trial is a conditionally Poisson process with a random rate.

Is there ``trial-to-trial variability'' in this example?  The answer depends on how we choose to {\em model} the data.  If our model class includes Cox processes, then we can view the trials as independent and {\em identically distributed}.  There is no non-stationarity, no ``trial-to-trial variability''.  Alternatively, if we wanted to model the data as a Poisson process, then the trials cannot be identically distributed, but must have different rates on different trials.  The data exhibit ``trial-to-trial variability'' from the point of view of our model class, and we would need to introduce trial-specific parameters to allow the Poisson intensity to change across trials.  In this example, ``trial-to-trial variability'' is not an intrinsic feature of the distribution generating the data, but rather appears as a consequence of our modeling assumptions.        

Of course, not all trial-to-trial variability results from a misspecified model.  If the first 5 trials are almost always 10 Hz and the next 5 are almost always 20 Hz, then the data cannot be both independent and identically distributed (iid) across trials, regardless of the model class.  We use the term ``trial-to-trial variability'' rather loosely in the main text.  Jitter methods are unaffected by coarse-temporal trial-to-trial variability, regardless of the source.
\end{leftbar}

\section{Spike re-sampling, spike jitter, and conditional inference}

\subsection{Trial shuffles and interval jitter}

\subsubsection*{Permutation tests}

Trial-shuffling corresponds to the statistical concept of a {\em permutation test}.  Given a sequence of data $x_1,\dotsc,x_n$, a permutation test is a test of the null hypothesis that there is nothing special about the {\em order} in which we observed $x_1,\dotsc,x_n$.  Monte Carlo permutation tests are easy to implement: simply randomly permute the order of the observed data, and do this many times, and check if the observed data looks unusual among the collection of permuted datasets.

There are a variety of interesting null hypotheses that can be tested with a permutation test.  For example, suppose that $x_1,\dotsc,x_m$ are iid samples from experimental condition 1 and $x_{m+1},\dotsc,x_n$ are iid samples from experimental condition 2.  If there is no difference in conditions, then there is nothing special about the ordering of the $x$'s, so a permutation test can be used to test for differences across conditions.

For another example, suppose that we have iid pairs $(x_1,y_1),\dotsc,(x_n,y_n)$ and we want to test whether $x_i$ and $y_i$ are independent.  If they are independent, then the ordering of the $x$'s does not matter, regardless of the ordering of the $y$'s.  So a permutation test can be used to test independence.  If $x_i$ and $y_i$ are simultaneously recorded spike trains, then this permutation test is the basis for ``shuffle correction'' in cross-correlation analysis.  In the text below we refer to trial shuffling as testing independence, but note that it is also testing that the trials are iid.  

\subsubsection*{Figure 1}

{\em Data generation.} On trial $k$, we first generate $\mu_{1,k},\dotsc,\mu_{40,k}$ iid uniform$(0,1)$.  Then we create the nonnegative function
\[ f_k(t) = \left(10 + \sum_{j=1}^{40} \sum_{\ell=-\infty}^{\infty} \frac{1}{2\sigma/\sqrt{2}}\exp\left(-\frac{|t+\ell-\mu_{j,k}|}{\sigma/\sqrt{2}}\right)\right)\ind\{t\in[0,1)\} \]
which satisfies $\int_0^1 f_k(t) dt = 50$ for any $\sigma > 0$, where $\ind\{A\}$ denotes the indicator (zero/one) function for the event $A$.  Each $f_k$ is the sum of a constant baseline function and $40$ different double-exponential probability density functions, which are wrapped around the interval $[0,1)$ to preserve the total integral.  For Figure 1 we use $\sigma = 0.05$.  Then we generate two independent observations from an inhomogeneous Poisson process with intensity function $f_k(t)$.  These are the two spike trains on trial $k$, one for each neuron.  (Later we create a third independent observation; the third spike train is used below in Figure 3.)  Each of the 100 trials is created in this manner (beginning with new random $\mu$'s).  This is a Cox process.  Figure 1A shows $f_1,\dotsc,f_5$ and the corresponding spikes on the first 5 trials.  Figure 1B, gray line, shows $\sum_{k=1}^{100} f_k(t)/100$.

We concatenated the trials to create a single long spike train for each of the spike trains.  (This is not important here --- we did it for algorithmic reasons --- but it does have tiny implications for edge effects at the trial boundaries.  It does not affect the validity of any of our statistical conclusions.  In particular, edge effects are not an issue for our interval jitter experiments below, because we placed interval boundaries at trial boundaries.)  Specifically, we mapped the spike times in trial $k$ from the interval $[0,1)$ to the interval $[k-1,k)$ by simply adding $k-1$ to the spike times.  For each neuron, this converts 100 trials of a one second spike train into a single observation of a 100 second spike train over the interval $[0,100)$.  Let $N_i$ be the total number of spikes from neuron $i$ and let $Y_{i,1}\leq Y_{i,2}\leq\dotsb\leq Y_{i,N_i}$ be the corresponding observed spike times after trial concatenation.  To recover the trial-relative times of a spike, we can use ${Y_{i,k}\bmod 1}$.

{\em PSTHs.} Peri-stimulus time histograms (PSTHs) were constructed using 50 ms box smoothing and averaging across trials.  In particular, the PSTH for neuron $i$ at trial-relative time $t\in(0,1)$ was
\[ \frac{1}{100}\sum_{k=1}^{N_i} \frac{\ind\{t-0.025 \leq (Y_{i,k}\bmod 1) < t+0.025\}}{0.050} \]
in units of Hz or spikes/s, where $\ind\{A\}$ is the indicator function of the event $A$, taking the value 1 if $A$ is true and 0 otherwise.  PSTHs were plotted using $t$ on a 2 ms grid as the two black lines in Figure 1B.

{\em CCHs and lag-0 synchronies.} Cross-correlation histograms (CCHs) were constructed using 2 ms box smoothing and considering lags up to $\pm 250$ ms.  The CCH at lag $\tau$ was 
\[ \sum_{k=1}^{N_1}\sum_{\ell=1}^{N_2} \ind\{\tau-.001 \leq Y_{2,\ell}-Y_{1,k} < \tau+.001\} \]
in units of total number of spike pairs across all trials.  CCHs were plotted for $\tau\in(.25,.25)$ on a 0.4 ms grid.  The CCH value at $\tau=0$ corresponds to the total number of $\pm 1$ ms precise synchronies.  Note that reversing the role of neuron $1$ and $2$ merely switches positive lags to negative lags (i.e., reflects the CCH horizontally about lag zero) and does not affect the definition of lag-0 synchronies.  The original CCH is shown in Figure 1C (black line).  The observed value of $c(0)$ is 575 (almost 6 synchronous pairs per trial).

{\em Shuffle-surrogate CCHs.}  A shuffle-surrogate CCH is created by randomly permuting the trial order for spike train 1 (i.e., before concatenating the spikes into a single trial), and then computing the CCH as above between the trial-shuffled version of spike train 1 and the original version of spike train 2.  We create a collection of $M=10,000$ such shuffle-surrogate CCHs, each using different independent realizations of random permutations of the trials.  We will denote the CCHs as $c_0, c_1,\dotsc, c_{M}$, where $c_0(\tau)$ is the original CCH at lag $\tau$ and $c_m(\tau)$, for $m > 0$, is the $m$th shuffle-surrogate CCH at lag $\tau$.  

{\em Lag-0 distribution.} The lag-0 distribution (Figure 1D) is simply a histogram of the values $c_0(0),\dotsc, c_{M}(0)$.  The black vertical line in Figure 1D shows where $c_0(0)$ occurs within this histogram.  It is well known that the right tail probability of this histogram is a p-value for testing the independence between the two spike trains, namely,
\[ \text{p-value} = \frac{1}{M+1}\sum_{m=0}^M \ind\{c_m(0)\geq c_0(0)\} \]  See the discussion of permutation tests above.

{\em Mean shuffle-surrogate CCH.} We also compute the empirical mean CCH after shuffling, that is
\[ \mu(\tau) = \frac{1}{M}\sum_{m=1}^M c_m(\tau) \]
The light gray horizontal curve in the middle of Figure 1C shows the mean shuffle-surrogate CCH as a function of lag.  Note that the average CCH after shuffling does not show the broad peak at zero.
The light gray vertical bar in the middle of the histogram in Figure 1D shows where $\mu(0)$ occurs within the lag-0 histogram.

{\em Shuffle-derived pointwise acceptance bands.}  To create the acceptance bands, we first sort the elements $c_0(\tau),\dotsc, c_{M}(\tau)$ to get a sequence $c_{(0)}(\tau)\leq \dotsb \leq  c_{(M)}(\tau)$.  Note that the indexing still starts at zero (corresponding to the minimum), which is not standard for order statistics.  Then we set $ a(\tau)= c_{(.025M)}$ and $ b(\tau)= c_{(.975M)}$.  The interval $[a(\tau), b(\tau)]$ contains (at least) 95\% of the shuffle-corrected CCH values at lag $\tau$ (including the original).  We repeat this separately for each $\tau$.  On Figure 1C the dark gray region corresponds to the interval between $[a(\tau),b(\tau)]$ as $\tau$ varies.  With $\tau$ chosen a priori (say, $\tau=0$), we can reject at level 0.05 the null hypothesis that the spike trains are independent whenever $c_0(\tau)$ does not land in the interval $[a(\tau), b(\tau)]$.  The dark gray region in Figure 1D corresponds to the part of the histogram between $a(0)$ and $b(0)$.    

{\em Controlling for multiple comparisons.}  If we do not pick $\tau$ a priori (or if we pick more than one $\tau$ a priori), and we look for some $\tau$ where $c_0(\tau)$ is outside of the pointwise acceptance interval, then we cannot reject at level 0.05 and must do something to control for multiple hypothesis tests.  Bonferonni is often too conservative.  Here is something simple that combines all $\tau$ in order to rigorously test independence that we have found works well in many situations and provides a nice visual display.  First, robustly standardize the collection of CCHs at each $\tau$ and then compute the maximum and minimum of each standardized CCH, that is, compute
\[ \begin{gathered}  \nu(\tau) = \frac{1}{M-1}\sum_{m=1}^{M-1}  c_{(m)}(\tau) \quad\quad\quad
 s(\tau) = \sqrt{\frac{1}{M-2}\sum_{m=1}^{M-1} \bigl( c_{(m)}(\tau)- \nu(\tau)\bigr)^2} \\ 
 c^*_m(\tau) = \frac{ c_m(\tau) -  \nu(\tau)}{ s(\tau)} \quad\quad\quad  c^+_m = \max_\tau  c^*_m(\tau) \quad\quad\quad  c^-_m = \min_\tau  c^*_m(\tau) \end{gathered} \]  
(Notice that $\nu(\tau)$ and $s(\tau)$ are computed without the extreme values, and are therefore robust to single outliers in either tail.  Notice also that they are based on the sorted values, and so may contain the original CCH.) 
Now order the maximums and minimums, i.e., $ c_{(0)}^+ \leq \dotsb \leq  c_{(M)}^+$ and $ c_{(0)}^- \leq \dotsb \leq  c_{(M)}^-$.  If $ c_{0}^* <  c_{(0.025M)}^-$ or $ c_{0}^* >  c_{(0.975M)}^+$, then we can reject the null hypothesis of independence at level 0.05.  Although it takes a while to wade through all of the preprocessing, it is straightforward to verify that the maximum (or minimum) statistics (which combine all time lags) can be shuffled without changing their distribution under the null hypothesis of independence.  The reader is referred to \cite{westfall1993resampling, nichols2003controlling} for more details about resampling based multiple testing.

We note that independence is rejected for the data in Figure 1 (as it should be --- the spike trains are correlated on coarse timescales by virtue of sharing the same random intensity function) using this particular method of controlling for multiple comparisons.

{\em Shuffle-derived simultaneous acceptance bands.}  The multiple comparisons method described above can be used to create {\em simultaneous} acceptance bands, which provide a nice visual display to augment pointwise bands.  They also show which time bin offsets were the most unusual from a multiple comparisons perspective.  Define $a^*(\tau)=c_{(0.025M)}^- s(\tau)+ \nu(\tau)$ and $b^*(\tau)=c_{(0.975M)}^+ s(\tau)+ \nu(\tau)$.  It is straightforward to verify that we reject the null hypothesis using the multiple comparisons method described above exactly when there exists a $\tau$ for which $c_0(\tau)$ is not in the interval $[a^*(\tau),b^*(\tau)]$.  On Figure 1C the light gray region corresponds to the interval between $[a^*(\tau),b^*(\tau)]$ as $\tau$ varies.  The CCH exceeds the simultaneous bands, signaling a rejection of the null hypothesis.  The light gray region in Figure 1D shows the points between $a^*(0)$ and $b^*(0)$ for easy comparison to Figure 1C, but note that the entire collection of CCHs, not just the lag-0 values, are used in the construction of simultaneous acceptance bands.

{\em Shuffle-corrected CCHs.}    For visualization purposes, we show the {\em shuffle-corrected CCH} as a black line in Figure 1E.  The shuffle-corrected CCH is the original CCH minus the shuffle mean, i.e., $c_0(\tau)-\mu(\tau)$.  Under the null hypothesis of independence, the shuffle-corrected CCH has expected value zero for all $\tau$, so ``large'' variations from zero (solid light gray line in Figure 1E) are evidence against the null.  The acceptance bands described above are one way to quantify ``large''.  On Figure 1E we also show the shuffle-corrected version of these bands, namely, the dark gray region corresponds to the interval between $[a(\tau)-\mu(\tau),b(\tau)-\mu(\tau)]$ as $\tau$ varies and the light gray region corresponds to the interval between $[a^*(\tau)-\mu(\tau),b^*(\tau)-\mu(\tau)]$ as $\tau$ varies.  Note that the shuffle-corrected CCH exceeds the shuffle-corrected bands exactly when the original CCH exceeds the original bands.  The ``correction'' is only for visualization purposes.

\subsubsection*{Interval jitter}

\begin{leftbar}
{\em Additional remarks.}  As with trial shuffling, interval jitter is a procedure for generating an ensemble of spike processes from a single observed process.  The difference is that spikes are jittered within a priori defined windows (``intervals'') rather than shuffled across trials.  Thus for a jitter window of $\delta$ ms, the recording is first partitioned into successive $\delta$ ms intervals. An ensemble of spike processes is then generated by relocating each spike, independently, to a random point within its original $\delta$ ms interval. Everything else is the same: a statistic, such as the number of one-millisecond synchronies, is chosen, and the observed value from the original, unperturbed, spike train is compared against the ensemble of values collected through the jitter process.  If many statistics are available, such as all of the lags in a CCH, then pointwise and simultaneous acceptance bands, and ``jitter-corrected'' displays can be constructed, as well. 
\end{leftbar}

\subsubsection*{Figure 2}

{\em Interval jitter.}  The $m$th interval-jitter surrogate for neuron $i$ using jitter windows of length $\delta>0$ is
\[ \bigl(Y^{(m)}_{i,1},\dotsc,Y^{(m)}_{i,N_i}\bigr) = \text{sort}\bigl(\delta\lfloor{Y_{i,1}/\delta}\rfloor+\delta U^{(m)}_{i,1}, \dotsc, \delta \lfloor{Y_{i,N_i}/\delta}\rfloor+\delta U^{(m)}_{i,N_i}\bigr) \]
where $U^{(m)}_{i,k}$ are iid uniform$(0,1)$ random variables for all $i,k,m$, and where $\lfloor{\cdot}\rfloor$ is the floor function (greatest integer less than).  In words, $\delta\lfloor{Y_{i,k}/\delta}\rfloor$ moves the $k$th spike of neuron $i$ to the start of its $\delta$-length jitter window.  Then we add $\delta U^{(m)}_{i,k}$ to get a uniformly chosen location within that same jitter window.  We do this independently for every spike of each spike train.  Finally, we sort the results of each spike train so that the spike times are increasing.   We used $\delta=0.02$ and repeated this $M=10,000$ times, i.e., $m=1,\dotsc,10,000$, to create $10,000$ surrogate jittered datasets.

{\em Replacing trial-shuffling with interval jitter.}  For each of the $M=10,000$ surrogate jittered datasets we computed the CCH between spike train 1 and 2, giving $M$ interval-jitter surrogate CCHs, say, $c_1,\dotsc,c_M$, using the same notation for the shuffle surrogate CCHs described above.  As before, it is convenient to let $c_0$ denote the CCH of the original (unjittered) dataset.  Now everything proceeds exactly as described above for the collection of shuffled-derived CCHs, except we use the collection of jitter-derived CCHs.  The lag-0 values, namely, $c_0(0),c_1(0),\dotsc,c_M(0)$, are still the $\pm 1$ ms synchronies, but now the surrogates correspond to the number of synchronies after jittering (instead of after shuffling).  P-values are computed the same way and acceptance bands are computed the same way.  The corresponding hypothesis test is described below.  For jitter-correction, we subtract the mean, say $\mu(\tau)$, of the jitter-surrogate CCHs from the original CCH and any acceptance bands for easy visualization.    

\subsection{Hypothesis Testing and Conditional Inference}

\begin{leftbar}
{\em Interval jitter null hypothesis.} Fix the jitter window length $\delta$.  Define
\[ S_i(k) = \#\bigr\{\text{spikes from neuron $i$ in time interval $[\delta(k-1),\delta k)$}\bigl\} \]
where time refers to absolute time, not trial-relative time.
For the data in Figure 1 with $\delta=0.02$, for example, $S_1$ and $S_2$ are each length $5000$ nonnegative integer-valued vectors.  The interval jitter null hypothesis is
\begin{itemize}
\item[$H_0$:] The conditional distribution of the data is {\em uniform} given the vectors $S_i$ for all $i$.
\end{itemize}
Notice that $H_0$ depends on $\delta$ via the definition of the $S$'s.  $H_0$ makes no assumption about the distribution of the $S$'s, which are coarse-temporal statistics of the spike trains.  
Therefore, $H_0$ says nothing about the comparison of two datasets with different interval counts.
On the other hand,
if two datasets have exactly identical $S$'s, even if the precise timings of spikes are completely different, then $H_0$ states that the two datasets are equally likely to have occurred.  

If (and only if) the $S$'s are independent Poisson random variables, then the null hypothesis states that the spike trains are independent inhomogeneous Poisson processes with piecewise constant intensity functions that are constant within jitter windows.  This is but one of infinitely many classes of point processes that are included in $H_0$.  For example, the null hypothesis also includes doubly-stochastic Poisson processes (Cox processes) whose random inhomogeneous intensity function is piecewise-constant over the jitter intervals, in which case the $S$'s are neither independent nor Poisson.

The version of interval jitter used here places jitter window boundaries at the points $\dotsc,-3\delta,-2\delta,-\delta,0,\delta,2\delta,3\delta,\dotsc$.  There is nothing special about using these particular jitter window boundaries.  Indeed, there is no reason that they need all be the same length.  The key for creating a valid hypothesis test is that the jitter window boundaries do not depend on the data and that the boundaries are in the same locations for all jitter surrogates of the same spike train.  Of course, the boundaries influence the interpretation of the resulting hypothesis test, since they create the null hypothesis.  
\end{leftbar}

There are many variations on this interval jitter theme.  A common one is to jitter only neuron 1 (say) and leave all of the other neurons fixed.  This tests hypotheses about the temporal resolution of neuron 1 with respect to the others, without imposing assumptions about the temporal resolution of the others.  The null hypothesis would be
\begin{itemize}
\item[$H_0$:] The conditional distribution of the data is {\em uniform} given the vector $S_1$, and given the values of all spike times other than those from neuron 1.
\end{itemize}
Since we are conditioning on more aspects of the data, this single neuron version of interval jitter is an even larger null hypothesis than the original.  Consequently, tests will tend to be more conservative (harder to reject), but a rejection here implies a rejection of the original null.  For computational reasons, we use the single neuron version of interval jitter below when we discuss relaxing the uniformity assumption in the null hypothesis. 

\subsection*{Conditional inference}

\begin{leftbar}
{\em Conditional versus unconditional inference.}  Jitter and many of the related resampling techniques currently in use in neuroscience are best interpreted as Monte Carlo conditional inference, or perhaps, Monte Carlo approximate conditional inference.  We have noted a systematic tendency to blur the distinction between conditional inference and unconditional inference with regard to these jitter-like resampling techniques.  Unconditional bootstrap, in particular, seems to be a common (mis)interpretation, perhaps because bootstrap computations are strongly associated with Monte Carlo methods.  Mistaking a conditional procedure for an unconditional one can lead to statistical and scientific errors.  This seems especially the case for the types of neuroscience examples that we have in mind.  In this section we provide a simple simulation example to illustrate the differences between conditional and unconditional inference.

Consider 100 trials, each lasting 1 second, of two simultaneously recorded spike trains.  We are interested in precise zero-lag correlations and want to understand the distribution of the statistic $T=c(0)$, defined above as the total number of $\pm 1$ ms synchronous spike pairs.  For example, we may want to understand the distribution of $T$ in order to construct accurate critical values for hypothesis testing.

For unconditional inference we want to understand the unconditional distribution of $T$.  For example, suppose that the neurons are independent homogeneous Poisson processes with rates 50 Hz (neuron 1) and 25 Hz (neuron 2), respectively, and that trials are iid.  We generated $10^5$ datasets from this distribution, computed the value of $T$ on each one, and displayed the resulting empirical distribution in Supplementary Figure S\ref{fig:jit} Panel A.  This is an excellent approximation of the true unconditional distribution of $T$ because we used so many Monte Carlo samples (from the true distribution).  If we did not know that the neurons were 50 Hz and 25 Hz homogeneous Poisson processes, but were given some data, we might try to use bootstrap to approximate this unconditional distribution of $T$.

For conditional inference we want to understand the conditional distribution of $T$ given some other event.  For interval jitter, this other event is the observed sequence of spike counts in all of the jitter windows for all neurons.  Supplementary Figure S\ref{fig:jit} Panels B--E show four different examples of conditional distributions of $T$ given different observed sequences of spike counts in 5 ms jitter intervals.  Again, each of these is actually an empirical distribution using $10^5$ Monte Carlo samples from the respective conditional distributions.  Sampling from the true conditional distribution is easy because the spike times are conditionally uniform given the spike counts.  For any given dataset, we will have a specific sequence of spike counts, which gives a specific conditional distribution of $T$, and this (and only this) is the distribution we want to understand.  

Notice that the conditional distributions can be quite different from the unconditional distribution.  Notice also that these conditional distributions are what interval jitter generates.  Finally, notice that a great many distributions (other than 50 Hz and 25 Hz independent homogeneous Poisson processes) give rise to the same conditional distributions of $T$.  Most of these processes will have markedly different unconditional distributions of $T$.  
\end{leftbar}

\begin{SCfigure}
\centering
\epsfig{file=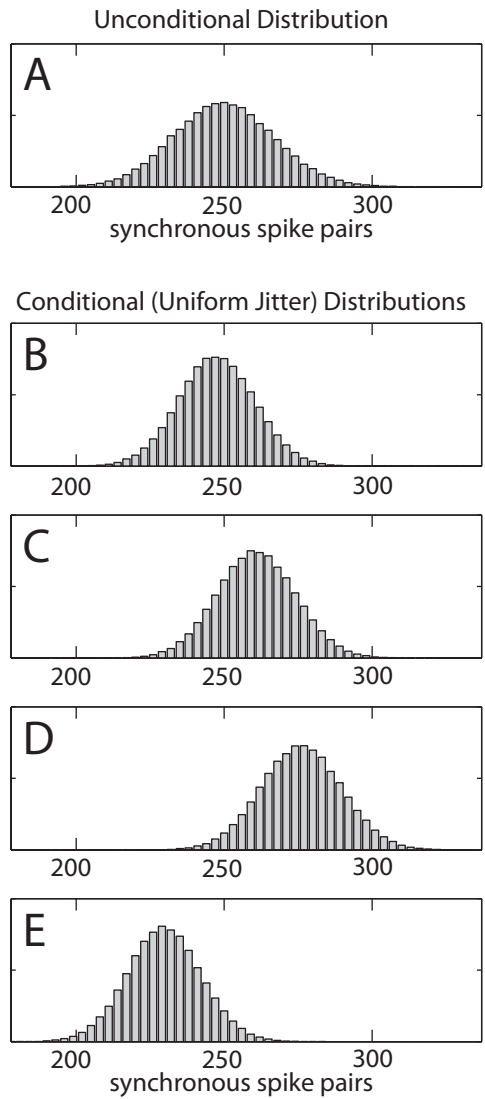}
\caption{ {\bf Unconditional versus conditional distributions.} Panel A shows the distribution of the total number of 1 ms synchronous pairs in 100 iid one-second trials between two independent 50 Hz and 25 Hz homogeneous Poisson processes.  This distribution was estimated using $10^5$ Monte Carlo observations.  The panels B--E correspond to four different observations, respectively, among the $10^5$ observations used to create panel A.  In each case, panels B--E show the conditional distribution of the total number of 1 ms synchronous pairs {\em given the sequence of spike counts in 5 ms windows}.  Each of these conditional distributions was approximated with $10^5$ uniform interval jitter Monte Carlo surrogates.  Since none of four exemplars have the identical sequence of spike counts, the conditional distributions are different.  All of the 5 graphs are on the same scale.  The observed synchrony counts for the original data from panels B--E were 225, 253, 274, and 235, respectively. \label{fig:jit}}
\end{SCfigure}

{\color{black}
\subsection{Synchrony as a Test Statistic}

Although any test statistic can be used in conjunction with jitter, we note that the choice of test statistic strongly affects both the statistical power of jitter-based hypothesis tests and the interpretation of any scientific conclusions drawn from jitter about the alternative hypothesis.
}

\subsection{Basic jitter}

{\em Basic jitter.}  Basic jitter is a heuristic procedure that jitters spikes in windows {\em centered at the original spikes}.  For basic jitter, the $m$th surrogate of spike train $i$ is
\[ \bigl(Y^{(m)}_{i,1},\dotsc,Y^{(m)}_{i,N_i}\bigr) = \text{sort}\bigl(Y_{i,1}+\delta (U^{(m)}_{i,1}-1/2), \dotsc, Y_{i,N_i}+\delta (U^{(m)}_{i,N_i}-1/2)\bigr) \]
Basic jitter does not correspond to Monte Carlo sampling from any conditional distribution.  If it did, then the conditioning event would have to specify (either explicitly or implicitly) the center of the jitter windows, which means it would specify the locations of the original spikes, which means the resulting conditional distribution would have no variability (it would be a point mass on the locations of the original spikes).  This is true for any type of spike-centered jitter, regardless of the distribution used to jitter spikes (here, uniform, but other authors have used Gaussian or triangular jitter distributions centered at the original spikes).  No such procedure should be interpreted as a statistical hypothesis test.

\subsection{Accidental synchrony}

\subsubsection*{Figure 3}

{\em Injected synchrony.} Locations for injected pairs were produced by sampling from the same inhomogeneous Poisson processes used in the original experiments depicted in Figure 1, but with the intensity functions scaled down from 50 Hz to 0.25, 0.50, and 0.75 Hz, respectively for the three additional experiments, corresponding to a mean total number of injected synchronous pairs of 25, 50, and 75 over the 100 trials. In order to maintain identical marginal (single-neuron) statistics following the addition of the synchronous pairs, each Poisson spike process was first thinned by subjecting every spike, independently, to random elimination. The actual probability of elimination was, for example, 0.005 for the
experiment with 0.25 Hz injected synchrony, and more generally was $h$ / 50 for the experiment with synchronies injected at $h$ Hz. 

Here are the exact details.  We refer to the data generation description above for Figure 1.  We begin with the identical dataset used in Figures 1 and 2.  We also create a third independent spike train from the same distribution, i.e., with intensity function $f_k(t)$ on trial $k$.  Let $Z_{i,k,j}$ denote the $j$th trial-relative spike time on the $k$th trial of neuron $i$, where $i=1,2,3$, and $k=1,\dotsc,100$, and $j=1,\dotsc,N_{i,k}$, where $N_{i,k}$ is the number of spikes from neuron $i$ on trial $k$.  The spike times (the $Z$'s) for these three spike trains are fixed for each of the different injected synchrony experiments.  We will use them to build datasets with differing amounts of synchrony.  

Now we generate iid uniform$(0,1)$ random variables $U_{i,j,k}$, one for each spike time $Z_{i,j,k}$.  These $U$'s are also fixed for all of the experiments.  To generate a dataset with $0 \leq h \leq 50$ Hz injected synchrony, we create two new spike trains as follows:
\begin{itemize}
\item The spike times in the new spike train 1 are those $Z_{1,j,k}$ with $U_{1,j,k} \leq 1-h/50$ and also those $Z_{3,j,k}$ with $U_{3,j,k} < h/50$.
\item The spike times in the new spike train 2 are those $Z_{2,j,k}$ with $U_{2,j,k} \leq 1-h/50$ and also those $Z_{3,j,k}$ with $U_{3,j,k} < h/50$.
\end{itemize}
Notice that the two new spike trains share the selected spikes from the third old spike train.  It turns out that each of the two new spike trains has the same distribution, individually, as those from the dataset in Figure 1.  (This is a result of the fact that thinned Poisson processes and superimposed independent Poisson processes are both still Poisson processes.)  When $h=0$, it is the same dataset and the two new spike trains are, of course, still independent.  But when $h > 0$ the two spike trains are no longer independent.  They share some spikes.  We use $h=0.0, 0.25, 0.50, 0.75$ in the experiments for Figure 3.  All other details are identical to the experiments in Figures 1 and 2.

\begin{leftbar}
The rationale for this method of creating an injected synchrony dataset is that the datasets for different amounts of injected synchrony are as similar as possible, and thus easier to compare.  The figures can, however, be misleading because (if they are interpreted as independent datasets) they give the impression of much smaller variability than would be actually observed across different experimental datasets with similar distributions.
\end{leftbar}

\begin{leftbar}
{\em Crude estimate of injected synchrony.}  The estimate of injected synchrony mentioned in Section 3.1 of the main text and in the Figure 3 legend is the height of the lag-0 peak in the jitter corrected CCH, namely, $c_0(0) - \mu(0)$, as discussed above in the Supplementary discussion of Figure 2.  We note that this estimate has substantial bias that can be improved by taking seriously the model of injected synchrony.  A small p-value should always be accompanied by some scientifically interpretable measure of the degree of departure from the null hypothesis.  This crude estimate serves just such a purpose for synchrony.  It is in units of total number of coincident spike pairs. Dividing by time gives an estimate of the injected synchrony rate.  An extremely small injected synchrony rate is unlikely to be scientifically interesting, regardless of how small the p-value is.    
\end{leftbar}

For convenience, the injected spike times are identical and could therefore be uniquely identified in the data.  Our test statistic, however does not take advantage of this.  The test statistic (millisecond-accurate synchrony count) detects 575 synchronous pairs even without any injected synchronies.  Indeed, the simulations and estimates are essentially unchanged if the injected spikes are not identical, but are offset by some random amount on the order of 1 ms.  This would be more realistic, but adds an additional layer of complication in the simulations, especially if we want to ensure that the marginal (i.e., single spike train) distributions are truly identical across differing amounts of injected synchrony.

\subsection{Temporal resolution}

\subsubsection*{Figure 4}

{\em Data distribution.}  We first generated $\mu_1,\dotsc,\mu_{40}$ iid uniform$(0,1)$.  The sorted values were 0.032, 0.034, 0.036, 0.046, 0.097, 0.098, 0.127, 0.142, 0.158, 0.171, 0.277, 0.278, 0.317, 0.392, 0.422, 0.485, 0.547, 0.632, 0.655, 0.656, 0.679, 0.695, 0.706, 0.743, 0.758, 0.792, 0.800, 0.815, 0.823, 0.849, 0.906, 0.913, 0.916, 0.934, 0.950, 0.957, 0.958, 0.959, 0.965, 0.971.  (These are the same $\mu$'s used on the first trial in the experiments in Figures 1--3.)  For the experiments here, these 40 values are fixed for all trials and all data sets.

For a fixed bandwidth $\sigma$, the 100 trials of each neuron are independent and they are each iid samples from an inhomogeneous Poisson process with intensity function
\[  f(t) = \left(10 + \sum_{j=1}^{40} \sum_{\ell=-\infty}^{\infty} \frac{1}{2\sigma/\sqrt{2}}\exp\left(-\frac{|t+\ell-\mu_{j}|}{\sigma/\sqrt{2}}\right)\right)\ind\{t\in[0,1)\} \]
This is the same as Figure 1, but here it is fixed for all trials (and the bandwidth is different for different experiments).  Recall that $\int_0^1 f(t) dt = 50$ for all $\sigma > 0$.  As $\sigma$ gets smaller, $f(t)$ becomes more concentrated around the $\mu_j$'s.  It approaches a constant of 50 as $\sigma$ gets very large, and it approaches a constant baseline of 10 with additional narrow spikes at each $\mu_j$ as $\sigma$ gets very small.  Our simulations do not venture to these extreme bandwidths, however.

Figure 4 uses bandwidths of $\sigma = 0.036, 0.020, 0.012, 0.008$, from top to bottom, respectively.

{\em Data generation.}  Much like Figure 3, the motivation for this sampling procedure is to make simulated datasets of different bandwidths be as similar as possible.  Each spike train is a superposition of many independent Poisson processes that combine together to give the final intensity function $f$ described above.  We will describe the process used to generate spike train 1 for each bandwidth.  The generation of spike train 2 proceeds identically and independently.  

The first piece corresponds to the 10 Hz baseline rate.  We generated 100 iid samples from a 10 Hz homogeneous Poisson process over $(0,1)$.  Each dataset (with different bandwidth) shares these spike times.  Then for each trial ($k=1,\dotsc,100$) and each $\mu_j$ ($j=1,\dotsc,40$) we generated iid Poisson random variables (not processes) with mean $1$.  Call these random variables $M_{k,j}$.  These $4000$ numbers are fixed and shared across all datasets.  For each $\sigma, k, j$ we generated $M_{k,j}$ iid samples from the probability density function (pdf)
\[   \left(\sum_{\ell=-\infty}^{\infty} \frac{1}{2\sigma/\sqrt{2}}\exp\left(-\frac{|t+\ell-\mu_{j}|}{\sigma/\sqrt{2}}\right)\right)\ind\{t\in[0,1)\}   \]
The point is that the number of spikes $M_{k,j}$ contributed by a given component $j=1,2,\dotsc,40$ to a given trial $k=1,2,\dotsc,100$ is the same for all bandwidths $\sigma$.
We did this by sampling from the double exponential distribution with pdf $\exp\bigl(|t-\mu_j|/(\sigma/\sqrt{2})\bigr)/\bigl(2\sigma/\sqrt{2}\bigr)$ and then taking the samples modulo 1 (which wraps samples outside of the unit interval back onto it).  For bandwidth $\sigma$ and trial $k$ the sorted collection of these samples and the original 10 Hz baseline spike times become the spike times on trial $k$ for neuron 1 in the bandwidth $\sigma$ dataset.

{\em Data processing.}  The data processing is identical to that in Figures 1--3.

\subsection{The jitter-corrected cross-correlation histogram}

Jitter-correction is described in the Figure 2 section above.  It is exactly analogous to shuffle-correction, except that it uses jitter, instead of trial-shuffling, to create the surrogates.  We note that the true expected value CCH after jittering can be computed exactly without using the empirical mean of jittered surrogates.  We do not describe the details here, since explicit creation of surrogates seems necessary to construct pointwise and simultaneous acceptance bands.

\section{Variations on the jitter theme}

\subsection{Rate of change of intensity functions}

\begin{leftbar}
Recall that if we consider only inhomogeneous Poisson processes, then interval jitter tests whether the time-varying intensity function is piecewise constant over jitter windows.  Sticking to the inhomogeneous Poisson model for motivation, a more realistic and interesting hypothesis about the time-varying intensity function would be that it is (essentially) piecewise linear, with a given bound on the percent change that can occur within any one jitter interval. Supplementary Figure S\ref{fig:piece} reproduces the first intensity function from the experiment discussed in Section 2.1, with a piecewise-constant approximation over 20 ms jitter intervals (panel A), and a piecewise-linear approximation, also over 20 ms intervals, but with the additional constraint that firing rate not change by more than 25\% within any one interval (panel B). Processes generated from the piecewise linear intensity function would be, for all intents and purposes, indistinguishable from processes generated from the original intensity function.

We provide some details below about how to modify interval jitter to allow for non-uniform jitter distributions, such as piecewise linear up to a certain amount of change as suggested above.
As with interval jitter, the null hypothesis is much more general than the set of Poisson processes with piecewise-linear intensity function. The null hypothesis is in terms of the conditional distribution on spike locations, given the numbers of spikes observed in the $\delta$ ms intervals. Conditioned on these numbers, the placements of spikes are assumed independent, with likelihoods that vary linearly within an interval and have a specified bound on change over the interval. (Uniform relocation is the special case in which the percent change within a jitter interval is bounded by zero.) The hypothesis can be tested in either member of a pair of recorded neurons. Spikes in the selected neuron are randomly relocated; the other neuron's spikes remain fixed. Synchrony is defined with respect to the second (``reference'') neuron: a spike in the selected neuron is synchronous if it falls within one ms of any of the (fixed) spikes in the reference neuron. As with interval jitter, the observed synchrony count is compared to the population of counts generated by the randomization process.

For each jitter interval of the selected neuron, the linear intensity function that maximally promotes synchronies is computed. These ``worst-case'' intensity functions define the spike-relocation distribution (``tilted jitter''). Formulas for computing worst-case intensities under the piecewise-linear hypothesis, and for various generalizations, are provided in Section B of the Mathematical Appendix. The point is that the resulting right-tail probabilities on synchrony counts are at least as big as the tail probabilities for any other local intensity functions that are consistent with the null hypothesis. Therefore, a test performed at a given level of significance under the single conditional hypothesis defined by these particular local intensity functions is simultaneously valid for any other
hypothesis in the compound null. In the search for evidence for fine-temporal structure it is conservative to resample spikes using these worst-case probabilities.
\end{leftbar}

\begin{SCfigure}
\centering
\epsfig{file=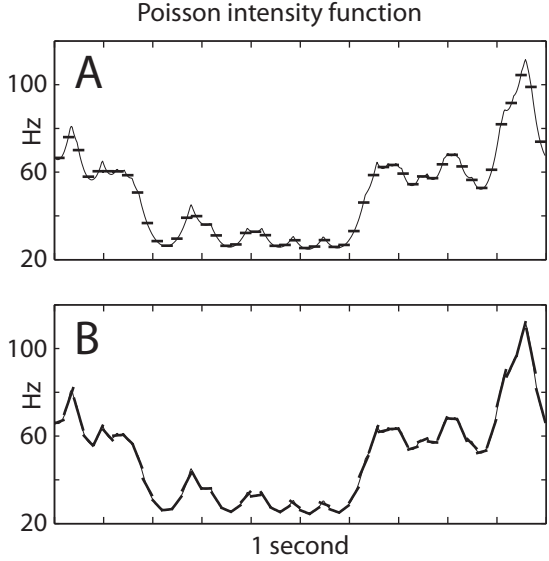}
\caption{ {\bf Piecewise approximations of a rate function.} The rate function shown in the upper-left panel of Figure 1A is approximated by a piecewise constant (panel A) and piecewise linear (panel B) function, with respect to a fixed 20 ms partitioning (``jitter intervals'').  The linear approximation is constrained to change by no more than 25\% within any 20 ms interval.  Subject to this constraint, the piecewise linear function shown in B best approximates the original.  Spike processes generated from the linear approximation would be nearly indistinguishable from processes generated from the original rate function. \label{fig:piece}}
\end{SCfigure}

{\em Non-uniform interval jitter null hypothesis.} There are many ways to relax uniformity in interval jitter.  Formally, we simply replace the word ``uniform'' in the specification of the null hypothesis with the description of another class of distributions that makes sense.  For example, we might replace ``uniform'' with ``a distribution that has at most 0.01 total variation distance from uniform''.  In practice, however, replacing the uniform distribution with some other class of distributions can create computational challenges.  The Mathematical Appendix discusses this in more detail.  Here we simply provide the details for the comments about non-uniform jitter in the main text and elsewhere in this supplement.

We experiment with replacing uniform with linear (see Proposition A.6 and Section B of the Mathematical Appendix).  For computational reasons, we fix the spike times of neuron 2 and only jitter the spikes of neuron 1.  In each jitter window (of each trial for neuron 1), given that there are $s$ spikes in that window, the null hypothesis dictates that these spike times are (conditionally) independent from the spike times in other windows and that the density on these $s$ spike times is iid with common pdf
\[ f_\theta(x) = \frac{1}{\delta}\left(1+\theta\biggl(\frac{x-a}{\delta}-\frac{1}{2}\right)\biggr)\ind\{a\leq x < a+\delta\} \quad \quad \quad \text{for some} \quad |\theta|\leq \frac{2\epsilon}{\epsilon+2} \]
where the jitter window is $[a,a+\delta)$ and where $\epsilon\geq 0$ is a bound on the allowable fraction of change in firing rate over any jitter interval.  The parameter $\theta$ is allowed to be different in each jitter window, but it cannot exceed the specified bound.  The strange looking form of the bound on $|\theta|$ comes from the fact that any such $\theta$ has
\[ \max_{x,y\in[a,a+\delta]} \frac{f_\theta(x)}{f_\theta(y)} - 1 \leq \epsilon \] 
The lefthand side ($\times$ 100\%) is the maximum percentage change in firing rate over the jitter interval $[a,a+\delta]$.  The null hypothesis constrains this to be less than $\epsilon$ ($\times$ 100\%).  
It is easy to verify that $f_\theta(x)\geq 0$ and integrates to one on $[a,a+\delta)$.
The null hypothesis has two parameters that control the scientific interpretation: the jitter window length ($\delta$) and the maximum fractional change in firing rate over any jitter window ($\epsilon$).  The case $\epsilon=0$ corresponds exactly to the original uniform null.  As $\epsilon$ increases, the null hypothesis enlarges to allow increasingly non-uniform processes.
In Section 3.1 of the main text, we mention allowing for various amounts of percentage change in firing rates within a jitter window.  These percentages refer to $\epsilon\times 100$\%.  

Also for computational reasons, we modify the synchrony test statistic for non-uniform jitter.  It is now (see the supplementary discussion for Figure 2)
\[ \sum_{i=1}^{N_1}\ind\left\{\min_{1\leq j \leq N_2} |Y_{1,i}-Y_{2,j}| \leq 0.001\right\} \]
which is the number of spikes in neuron 1 that participate in a synchronous pair (instead of the total number of synchronous pairs --- these two measures tend to be quite similar).  For this test statistic, it is computationally straightforward to choose the $\theta$ in each jitter window that creates the most amount of synchrony while satisfying the chosen $\epsilon$ bound (see Section B of the Mathematical Appendix).  Then we can jitter using the corresponding $f_\theta$ in each jitter window (instead of uniform jitter).  The resulting jitter distribution of synchrony is maximally conservative, meaning that we get a valid p-value for the non-uniform null hypothesis.

\subsection{Re-sampling patterns}

{\em Patterns.} Fix a parameter $R\geq 0$.  Consider a spike train $Y_1\leq \dotsb \leq Y_N$.  Define $Y_0=-\infty$ and $Y_{N+1}=\infty$ for notational convenience.  We can uniquely partition the spike train into patterns as follows: $Y_j,\dotsc,Y_k$ is a pattern if and only if 
\begin{equation}
\begin{aligned} & Y_j-Y_{j-1} > R & & \text{and}\\
& Y_{k+1}-Y_k > R & & \text{and} \\
& Y_{i+1}-Y_i \leq R  & & \text{whenever $j\leq i < k$}
\end{aligned}
\end{equation}
where $1\leq j\leq k \leq N$, in which case $Y_j$ is the starting time of the pattern.  
With a slight abuse of terminology, we will say that two patterns are the same if they have the same number of spikes and identical sequences of interspike intervals.
Two spike trains $Y_1\leq \dotsb \leq Y_N$ and $Y'_1\leq \dotsb \leq Y'_{N'}$ have the same sequence of patterns if and only if
\begin{equation}
\begin{aligned}
& N = N' & & \text{and} \\
& Y_{i+1}-Y_i = Y'_{i+1}-Y'_i & & \text{whenever $Y_{i+1}-Y_i \leq R$, and} \\
& Y'_{i+1}-Y'_i > R & & \text{whenever $Y_{i+1}-Y_i > R$}
\end{aligned}
\end{equation}
where $1\leq i < N$.  Notice that $Y$ and $Y'$ have the same patterns if and only if each spike in $Y$ has the same length $R$ history as the corresponding spike in $Y'$, where the length $R$ history of a spike $Y_i$ is the (perhaps empty) list of all spike times relative to $Y_i$ that occur in the $R$-length interval immediately preceding $Y_i$.  If there are 3 spikes in this interval, then the history would be $(Y_{i-3}-Y_i,Y_{i-2}-Y_i,Y_{i-1}-Y_i)$.

{\em The pattern-encoding statistic.} If $Y_{i,1}\leq \dotsb \leq Y_{i,N_i}$ is the spike train from neuron $i$, it is straightforward to verify that the following $N_i\times 2$ matrix only encodes the sequence of patterns and the sequence of length $\delta$ interval jitter windows that contain the start of each pattern: 
\[ V_i(\ell,1) = \ind\{Y_{i,\ell}-Y_{i,\ell-1} > R\} \quad \quad  \text{and} \quad \quad V_i(\ell,2) = \begin{cases} \lfloor{Y_{i,\ell}/\delta}\rfloor & \text{if $V_i(\ell,1)=1$} \\
Y_{i,\ell}-Y_{i,\ell-1} & \text{if $V_i(\ell,1)=0$} \end{cases} \]
for $\ell=1,\dotsc,N_i$, where we define $Y_{i,0}=-\infty$ for notational convenience.
The first column of $V_i$ indicates, with a one as opposed to a zero, which spikes begin a pattern.  The second column of $V_i$ gives the jitter window for spikes that begin a pattern and gives the preceding interspike interval for spikes that are within a pattern.

{\em The pattern jitter null.}  Fixing the history parameter $R$ and the jitter window length $\delta$, the pattern jitter null is
\begin{itemize}
\item[$H_0$:] The conditional distribution of the data is {\em uniform} given the matrices $V_i$ for all $i$.
\end{itemize}

{\em Monte Carlo pattern jitter.}  A fast algorithm for sampling from the pattern jitter null distribution is described in \cite{harrison-geman}.  We use the discretized version of the algorithm below, which constrains all spike times to a fine grid.  In many real applications, spike times are recorded at some fixed precision, and jitter surrogates should also be constrained to this same precision.  Monte Carlo pattern jitter surrogates can be used just like trial-shuffled surrogates or interval jitter surrogates to create p-values, acceptance bands, and visual ``corrections'' of CCHs or other graphical displays.

\subsubsection*{Figure 5}

The observed spikes (top row of each subfigure) are the spike times from neuron N1 (see Section 4) recorded between 3 and 4 seconds after the beginning of the experiment (a few seconds before the beginning of the first trial).  They were selected for illustration purposes.  The spikes were discretized at $1/30$th ms (30000 bins per second) and the discretized version of the pattern jitter sampling algorithm was used to create the example surrogates for different combinations of $\delta$ and $R$.

{\color{black} 
\subsubsection*{Pattern jitter for bursting neurons}

Here we experiment with a modification of the data sets used in Figure 3 (see Section 2.5) to illustrate the utility of pattern jitter for non-Poisson neurons.  In particular, we do a cross-correlation analysis of strongly bursting (artificial) neurons and show that pattern jitter is important in this context for preserving the nominal level of hypothesis tests.  This is interesting because the {\em expected} number of synchronies should be largely unaffected by the auto-correlation structure of individual neurons.  Nevertheless, the {\em distribution} is affected and so the auto-correlation must be accounted for when designing proper hypothesis tests.

We begin with the identical data set used for Figures 1 and 2.  The notation that we use here is described above in Section 2.5 for Figure 3.  Define $d_{i,k}=\lfloor{N_{i,k}/3}\rfloor$ to be the number of spikes in trial $k$ for spike train $i$ divided by 3 and rounded down to the nearest integer.  For each $i=1,2$ and $k=1,\dotsc,100$, we remove $2d_{i,k}$ spikes uniformly at random from that trial and that neuron.  Then we choose $d_{i,k}$ of the remaining spikes (again, chosen uniformly at random; this will be all of the remaining spikes if $d_{i,k}$ is a multiple of 3) and make them bursts of 3 spikes: we leave the original spike; we add a new spike uniformly in the interval $(8,9)$ ms after the original spike; we add another new spike uniformly in the interval $(16,17)$ ms after the original spike.  In the event that one of the new spikes lands outside of the trial interval, we repeat the entire process for that trial of that neuron.  This procedure takes the original data and makes the neurons burst (with high probability) in a succession of three spikes separated by about 8 ms per spike, while leaving the total number of spikes on each trial the same, and while leaving the time-varying, trial-varying intensity of the process largely unchanged.  These are highly non-Poisson spike trains.    

\begin{figure}
\centering
\epsfig{file=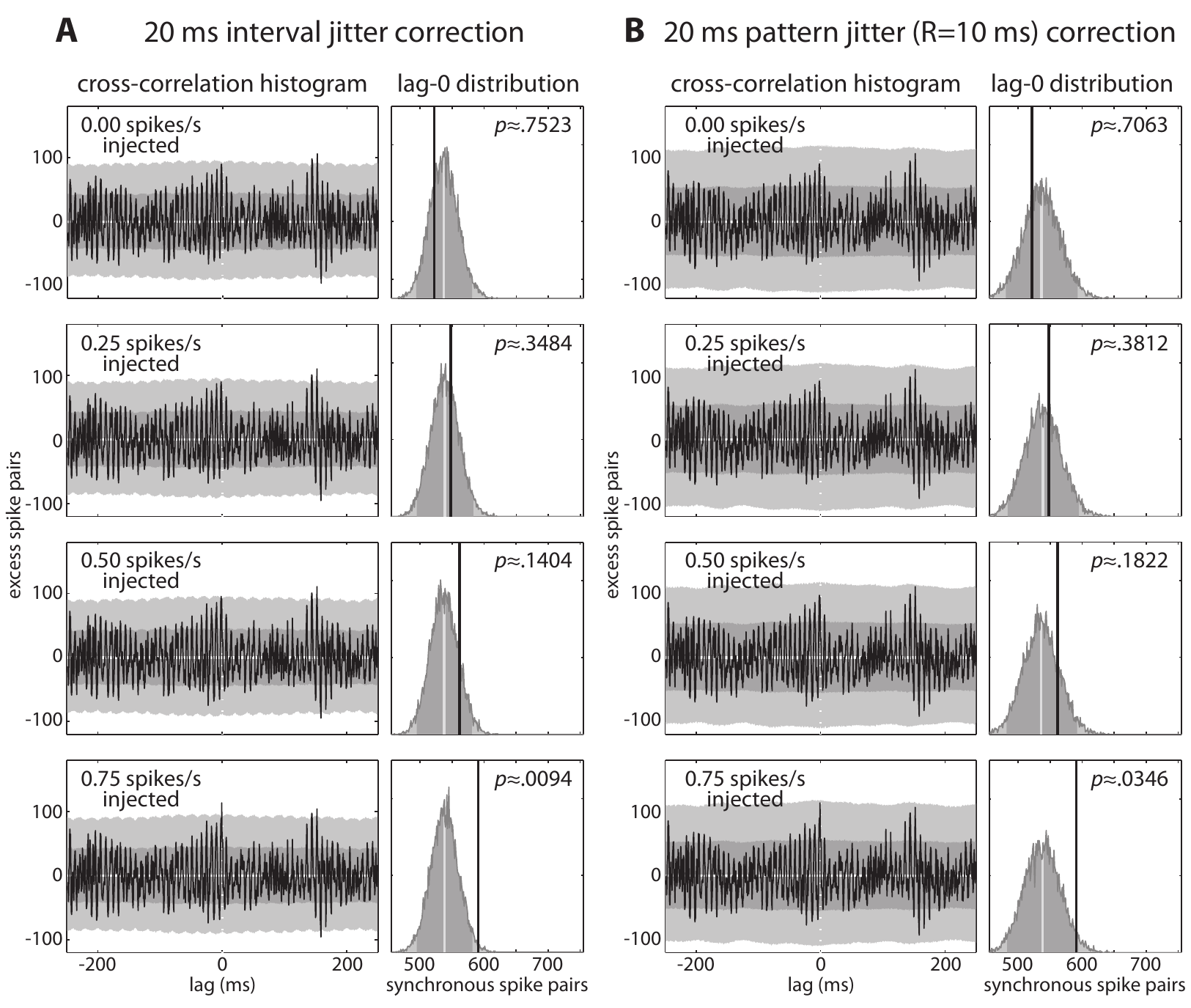}
\caption{ {\bf Interval versus pattern jitter corrected CCHs for bursting neurons.} Compare with Figure 3 in the main text.  The data set for this figure is described in the Supplementary section {\em Pattern jitter for bursting neurons}. \label{fig:patjit-sync}}
\end{figure}

Beginning with this as the baseline dataset for spike trains 1 and 2, and using the same spike train 3, we repeat the thinning/injection procedure of Figure 3 to create datasets with varying degrees of zero-lag injected synchrony.  Then we repeat the analysis of Figure 3, both with $\delta = 20$ ms interval jitter (Figure \ref{fig:patjit-sync}A) and with $\delta=20$ ms and $R=10$ ms pattern jitter (Figure \ref{fig:patjit-sync}B).  The simultaneous acceptance bands are violated with zero injected synchrony using interval jitter.  This is a correct rejection --- the bursting is a type of fine-temporal structure --- but it is not what we normally want with a cross-correlation analysis.  Implicit in a cross-correlation analysis is the understanding that the resulting conclusions are about the statistical relationship between two neurons.  We wanted to see synchrony (or a lack thereof) without artifacts introduced by the auto-correlation structure of the spike trains.  Pattern jitter, by preserving the bursts in the resamples, correctly calibrates the acceptance bands and does not report a rejection for the zero injected case.  Again, this is a correct conclusion for pattern jitter, because pattern jitter is a different null hypothesis that allows for certain types of auto-correlation structure (including the types of bursting introduced here).  Unlike interval jitter, however, the correct conclusion from pattern jitter corresponds with the standard scientific interpretation of significant correlation.  

As we add injected synchronies, pattern jitter behaves similarly to interval jitter from Figure 3 and begins to flag as significant the injected zero-lag synchronies (as it should --- injected synchronies are not in the pattern jitter null hypothesis).  Statistically speaking, pattern jitter has similar power as interval jitter for detecting injected synchronies in this case.

There are, of course, much more striking examples for illustrating the importance of pattern jitter, but the role of pattern jitter in these examples tends to be so obvious that detailed simulations are superfluous.  In particular, if the test-statistic depends highly on the auto-correlation structure of individual spike trains, such as the repeated spiking motifs of Section 4.1, and if the short-range auto-correlation structure is not of scientific interest, such as refractory periods and bursting, then pattern jitter, or a similar null hypothesis, is obviously important in order to focus inference on the features of interest, while ignoring the short-range auto-correlations of non-interest.  See \cite{oram99} for examples.
}

\section{Jitter analysis of three motor-cortical neurons}

\subsubsection*{Neurophysiological methods}

Three neurons (designated N1, N2, and N3) were selected from simultaneous multi-neuronal recordings made in primary motor cortex (MI) of the monkey. Single-unit activity was recorded from a multi-electrode array composed of 100 electrodes (1.0 mm electrode length; 400 $\mu$m inter-electrode separation) that was chronically implanted in the arm area of MI of one macaque monkey (Macaca Mulatto) --- see \cite{nicho2007} for details. During a recording session, signals from up to 96 electrodes were amplified (gain, 5000), band-pass filtered between 0.3 Hz and 7.5 kHz, and recorded digitally (14-bit) at 30 kHz per channel using a Cerebus acquisition system (Cyberkinetics Neurotechnology Systems, Inc., MA). Only waveforms (1.6 ms in duration; 48 sample time points per waveform) that crossed a threshold were stored and spike-sorted using Offline Sorter (Plexon, Inc., Dallas, TX). The monkey was operantly trained to move a cursor appearing above the monkey's hand location to targets projected onto a horizontal, reflective surface in front of the monkey. At any one time, a single target appeared at a random location in the workspace, and the monkey was required to move to it. As soon as the cursor reached the target, the target disappeared and a new target appeared in a new, random location. All of the surgical and behavioral procedures were approved by the University of Chicago�s IACUC and conform to the principles outlined in the {\em Guide for the Care and Use of Laboratory Animals} (NIH publication no. 86-23, revised 1985).

\subsubsection*{Figure 6}

Spike times from the first second of each successful trial were selected.  This results in 3 simultaneously recorded spike trains, each with 391 one-second trials.  Spike times in this data set are discretized to $1/30$th ms.  PSTHs and CCHs are computed identically as in Figure 1 (except with 391 trials instead of 100 trials).  Trial shuffling and the resulting data processing are also identical to Figure 1.

{\em Interval jitter.} For interval jitter (middle column) we partitioned each trial (i.e., the first second of each trial) into 50 jitter windows of length 20 ms each.  (This implies that our jitter analysis is also conditional on the starting times of successfully completed trials, since we used this information to select the data.)  Each 20 ms jitter window has 600 time bins of $1/30$th ms each.  We jittered each spike by uniformly and independently choosing one of the 600 bins {\em without replacement} in its respective jitter window.  (The modeling assumption here is that a time bin cannot have multiple spikes, which is realistic for such small time bins.  If the observed data has no time bins with multiple spikes, then sampling without replacement agrees with $R=0$ pattern jitter.)  

{\em Pattern jitter.} For pattern jitter (right column) we proceeded as in interval jitter, but used Monte Carlo samples from the $R=0.1$ seconds pattern jitter null.  Spike times outside of the first second of each trial were ignored.  (In practice, a more careful procedure would be to condition on these ignored spike times --- or better: pattern jitter the entire recording --- so that patterns extending outside of the trial boundaries are still preserved.  The differences are negligible in this dataset using the synchrony test statistic.  We ignored spikes outside of the first second so that the dataset would remain in the same format as our simulation experiments.) 

\subsection{Beyond synchrony: Other measures of precision}

\subsubsection*{Figure 7}

{\em Maximum repeating triplets test statistic.}  Let $Y_1\leq \dotsb\leq Y_N$ be the sequence of spike times discretized to $1/30$th ms and {\em measured in units of milliseconds}, so that $\text{round}(Y_\ell-Y_k)$ is the closest integer number of milliseconds that elapse between the $k$th and $\ell$th spike.  For each $i,j\in\{0,1,\dotsc,1000\}$, we first compute the number $H(i,j)$ defined as
\[ \begin{aligned} H(i,j) & = \sum_{k=1}^{N-2}\sum_{\ell=k+1}^{N-1}\sum_{m=\ell+1}^{N} \ind\{\text{round}(Y_\ell-Y_k)=i\}\ind\{\text{round}(Y_m-Y_\ell)=j\} \\ & \quad \times \ind\{Y_k,Y_\ell,Y_m \text{ all occur during the same successfully completely trial} \} \end{aligned} \]
which is the total number of three (not necessarily consecutive) increasing spike times such that the sequence of two intervals between the spike times (rounded to the nearest millisecond) is $(i,j)$ ms and such that all of the spike times come from the same successfully completed trial.  For computational simplicity, we only consider intervals up to 1000 ms.  Each $(i,j)$ pair defines a ``triplet''.  $H(i,j)$ counts how many times triplet $(i,j)$ repeats in the data.  (The data had an absolute refractory period of more than $1/2$ ms, so $H(i,j)=0$ whenever $i=0$ or $j=0$.)

The test statistic used in Figure 7 is the maximum number of repeating triplets, namely,
\[ \max_{i,j\in\{0,1\dotsc,1000\}} H(i,j) \]
Notice that pattern jitter with $R \geq d+0.5$ milliseconds exactly preserves the value of $H(i,j)$ for $i,j\leq d$.

{\em Pattern jitter.}  We generated surrogates from various pattern jitter null hypotheses with differing combinations of $R$ (history length) and $\delta$ (jitter window width) as before.  We conditioned on all spike times that were not part of a successfully completed trial, meaning that we held these spike times constant and included them in the definition of a pattern (but they were not included in the test statistic).  In fact, we also conditioned on the first and last spike time within in trial, to better ensure that there was no unexpected interaction between our choice of trial boundaries and the test statistic.  It turns out that none of the additional conditioning matters in the experiments here, but we feel that it is always good practice to err on the side of caution, especially when using complicated test statistics like the maximum number of repeating triplets.  

One can, of course, over-condition to the point that the surrogates do not have enough variability to draw meaningful conclusions (loss of power).  For our dataset the maximum value of $H$ was achieved with $(i,j)=(23,27)$, so that $R\geq 27.5$ ensures that $\max H$ can only increase with jittering.  We can thus never reject the null with $R\geq 27.5$ using this test statistic (and right-sided rejection regions).  Hence, choosing $R\geq 27.5$ certainly results in over-conditioning, and the problem likely begins for somewhat smaller $R$, as well.

\subsection{Temporal resolution and the neural representation of behavior}


\subsubsection*{Likelihood ratio tuning curves}  

\begin{leftbar}
One motivation for using likelihood ratio tuning curves 
\[ \theta(d) = \frac{P(D=d,S=1)}{P(D=d)P(S=1)} = \frac{P(D=d|S=1)}{P(D=d)} = \frac{P(S=1|D=d)}{P(S=1)} \] is the observation that $\theta(d)$ is proportional to the more familiar tuning curve $P(S=1|D=d)$ through multiplication by $P(S=1)$.  Multiplication by $P(S=1)$ conflates synchrony, per say, with directional tuning, and therefore complicates the interpretation of jitter or other resampling methods that might be used to calibrate the relationship between synchrony and direction, or, more generally, between fine-temporal events and other measures of behavior.

A second way to look at the role of likelihood ratios is to think of directional tuning as a collection of binary classification problems, one such problem for every direction $d$. Fix a single direction $d$ and consider the problem of guessing whether or not a certain event has occurred given an observed movement in that direction.  The ``event'' could be a spike in a particular neuron or the occurrence of a synchrony across two neurons, at about 100 ms (say) before the observed movement. Formally, the problem is to guess between $S=1$ and $S=0$ when the event of interest is a synchrony of spikes.  In both cases, the evidence is the movement direction, $D=d$.  These are binary classification problems and by the Neyman-Pearson lemma \cite{neyman1933problem,lehmann2005testing} 
the optimal classifier, in terms of the best tradeoff between sensitivity and specificity, is achieved by comparing the likelihood ratio, $\theta(d)$, to a threshold.  A particular threshold fixes a particular sensitivity/specificity pair; the set of all thresholds defines the optimal ROC (``receiver operating characteristic'') curve. 

For the reader more familiar with classification based on the ratio $P(D=d|S=1)/P(D=d|S=0)$, we note that 
\[ \theta(d) = \frac{P(D=d|S=1)}{P(D=d)} = \varphi\biggl( \frac{P(D=d|S=1)}{P(D=d|S=0)} \biggr) \] where $\varphi$ is the monotone function $\varphi(x)=x/\bigl(P(S=0)+P(S=1)x\bigr)$.  Hence the ratios are equivalent --- change the threshold and get the identical classifier.   
\end{leftbar}

\subsubsection*{Figure 8}

{\em Data preprocessing.}  Time was discretized to 1 ms bins.  The $(x,y)$ position of the hand (cursor) was smoothed with a Gaussian kernel using a bandwidth (standard deviation) of 25 ms and also shifted back in time 100 ms.  The offset of $100$ ms was chosen to approximate the delay between motor cortex and observed kinematics.  We will use $(x(t),y(t))$ to denote this smoothed and shifted position vector at time bin $t$.    The movement direction $D(t)$ at time bin $t$ was the angle of the vector $(x(t+1)-x(t-1),y(t+1)-y(t-1))$.  

For each neuron $i=1,2,3$, let $Y_{i,1},\dotsc,Y_{i,N_i}$ denote the spike times of neuron $i$.  For each pair of neurons $i,j$, we use $S_{i,j}(t)$ to denote the zero-one process of spike or no spike from neuron $i$ in time bin $t$ that happens to be within 1 ms of a spike from neuron $j$, namely,
\[ S_{i,j}(t) =\ind\bigl\{\text{$Y_{i,\ell}\in$ time bin $t$ and $|Y_{i,\ell}-Y_{j,k}|\leq 1$ ms for some $\ell=1,\dotsc,N_i$ and $k=1,\dotsc,N_j$}\bigr\} \]

When jittering, we jitter spikes before the 1 ms discretization step (i.e., using the original $1/30$th ms time bins), and then re-compute the $S$ processes (which are discretized at 1 ms bins).

{\em Estimating the likelihood ratio tuning curves.}  We abuse notation and use $P(D=d)$, which typically means the probability that $D$ equals $d$, when, in fact, we mean the probability density of $D$ at $d$.  The same comment applies for related quantities, like $P(D=d|S=1)$.  Note however, that $S$ is discrete, so $P(S=1)$ and related quantities are actual probabilities.

We estimate the density $P(D=d)$ using the list of values 
\[ \mathcal{D} = \bigl(D(t):\text{$t$ is within a successfully completed trial}\bigr) \]  We use kernel density estimation applied to the list of values in $\mathcal{D}$ with a Gaussian kernel with bandwidth $\sigma=0.1$ radians (which is approximately 6 degrees) that wraps around the interval $[-\pi,\pi)$.  Wrapping enforces the natural periodicity in $P(D=d)$ and prevents unnatural tapering for $d$ near $\pm\pi$.  The formula for our estimator is
\[ \widehat P(D=d) = \frac{1}{|\mathcal{D}|}\sum_{u\in\mathcal{D}} \sum_{\ell=-\infty}^\infty\frac{1}{\sigma}\phi((d-u+2\pi\ell)/\sigma) \ind\{d\in[-\pi,\pi)\} \]
where $\phi(z)=\exp(-z^2/2)/\sqrt{2\pi}$ is the standard normal density function.  This may look complicated, but except for $d$ near $\pm\pi$, it is virtually identical to the standard Gaussian kernel density estimator
\[ \frac{1}{|\mathcal{D}|}\sum_{u\in\mathcal{D}} \frac{1}{\sigma}\phi((d-u)/\sigma) \ind\{d\in[-\pi,\pi)\} \]
In fact, since our bandwidth is small, only $\ell=-1,0,1$ contribute to the sum in the formula for $\widehat P(D=d)$, and a simple trick to compute $\widehat P(D=d)$ is to use standard Gaussian kernel density estimation with the dataset $\mathcal{D} \cup (\mathcal{D}-2\pi) \cup (\mathcal{D}+2\pi)$ (and then multiply the resulting estimator by 3).  Note that $\mathcal{D}$ is a list of values, meaning that it can have repeat values.

We estimate $P(D=d|S_{i,j}=1)$ using the list of values
\[ \mathcal{D}_{i,j} = \bigl(D(t):\text{$t$ is within a successfully completed trial and $S_{i,j}(t)=1$}\bigr) \]  
again with Gaussian kernel density estimation with bandwidth $0.1$ radians and with wrapping around $[-\pi,\pi)$. 

Our final estimates of the likelihood ratio tuning curves are the ratios of these kernel density estimators, namely
\[ \widehat \theta_{i,j}(d) = \frac{\widehat P(D=d|S_{i,j}=1)}{\widehat P(D=d)} \]
Figure 8 shows $\widehat \theta_{1,2}$, $\widehat \theta_{1,3}$, and $\widehat \theta_{2,3}$, respectively, from left to right.

{\em Pattern jitter.} We use the identical pattern jitter procedure from Figure 7, but we independently jitter each of the three spike trains, not just N2.  For each jittered dataset, we compute new likelihood ratio estimators $\widehat \theta_{i,j}$ according to the identical procedure.  These estimators can potentially change a lot, because jittering can completely change which spikes are synchronous.  We created 1000 independent jittered datasets leading to $\widehat \theta_{i,j}^{0},\widehat \theta_{i,j}^{1},\dotsc,\widehat \theta_{i,j}^{1000}$, where $\widehat \theta_{i,j}^{0}$ comes from the original dataset and $\widehat \theta_{i,j}^{k}$ comes from the $k$th jitter-surrogate dataset.

{\em Acceptance bands and p-values.}  For the purposes of a test-statistic for hypothesis testing or for the purposes of constructing acceptance bands, there is no fundamental distinction between a CCH, say $c(\tau)$, and an estimator of a tuning curve, say $\widehat\theta(d)$.  We constructed acceptance bands for $\widehat\theta$ exactly as for the CCH, with one difference.  For the simultaneous acceptance bands, we first transformed $\widehat\theta$ to a logarithmic scale, namely, $\log\widehat\theta$, to remove some of the inherent asymmetry between large and small values of the likelihood ratio (which makes the standardization using $\nu$ and $s$ for the simultaneous bands behave better).  We computed simultaneous acceptance bands for $\log\widehat\theta$ and then transformed the bands back to the appropriate scale with an exponential operation.  Also, for Figure 8, we do not apply any ``correction'' (i.e., subtracting the mean of the surrogates to aid visualization) because we want to emphasize the the tuning curves vary strongly with direction.

\begin{leftbar}
The pattern jitter null hypothesis was not rejected using the simultaneous test for any of the three neuron pairs.  Although the estimated tuning curves do, in fact, exceed the pointwise acceptance bands at a few directions, we had no way to identify these directions a priori before seeing the data, so the simultaneous test is more appropriate.  (The situation is somewhat different for CCHs, because lag-0, corresponding to precise synchrony, can be identified as special a priori, and pointwise rejection at lag-0 is better justified.)  Even under the null hypothesis there should be some excursions outside of the pointwise bands (which we see, as expected), because each corresponds to a separate test of the same null hypothesis.  For the simultaneous bands, however, under the null hypothesis, the probability that the {\em entire} tuning curve falls within the light-grey bands is at least 0.95.  We do not expect to see any excursions outside of these bands.
\end{leftbar}

\section{Summary and Discussion}

{\color{black} Additional references can be found in the references section of this Supplement.  The list is not meant to be exhaustive.}

\newpage

\section*{Supplementary References}

Supplementary references are broken into categories in the following order: {\it On trial-to-trial variability; On jitter: methodology; On jitter: applications to neurophysiology; On spike timing; Statistical analysis of spike trains; Miscellaneous.} For topics of broad scope, we focus on prominent and/or representative selections that provide context for the main text.

\bibliographystyle{amsplain}

\begin{btSect}{statres_trial}
\subsection*{SR 1. On trial-to-trial variability}
\btPrintAll
\end{btSect}

\begin{btSect}{statres_jitter_meth}
\subsection*{SR 2. On jitter: methodology}
\btPrintAll
\end{btSect}

\begin{btSect}{statres_jitter_app}
\subsection*{SR 3. On jitter: applications to neurophysiology}
\btPrintAll
\end{btSect}

\begin{btSect}{statres_timing}
\subsection*{SR 4. On spike timing}
\btPrintAll
\end{btSect}

\begin{btSect}{statres_statistics}
\subsection*{SR 5. Statistical analysis of spike trains}
\btPrintAll
\end{btSect}

\begin{btSect}{statres}
\subsection*{SR 6. Miscellaneous}
\btPrintAll
\end{btSect}

\newpage

\appendix

\section*{Mathematical Appendix} 

In the main text, we illustrated the idea and purpose of jitter largely through a collection of spike-resampling algorithms, and for pedagogical reasons we connected the resampling algorithms to statistical concepts in words. In parallel, statistical methodologies, if properly formulated, can be used as tools for making mathematically-precise, probabilistic, statements. Here we provide this foundational view by placing those resampling algorithms used in the main text in correspondence with the statistical models, and null hypotheses, with which they are associated.

In statistical terms, the approach may broadly be understood as an instance of {\it conditional inference}. Null hypotheses are specified via conditional distributions on the data. Classes of conditional distributions express, as null hypotheses, mathematical models of various scientific questions concerning spike timing. The null hypotheses justify the resampling algorithms and related methods of inference.

The key ideas are present in the simplest case: discrete time, single neuron, uniform interval jitter.  We begin with a brief description of that special case, and then proceed with a systematic development of underlying mathematical principles.

A finely discretized spike train can be represented as a binary vector, say $\bs{x}=(x_1,\dotsc,x_m)$, where $x_i=1$ if there is a spike in time bin $i$, and $x_i=0$ otherwise.  Fix a jitter window width, say $\Delta$ time bins, and let $N_k(\bs{x})$ be the number of spikes in the time bins $\{(k-1)\Delta+1,\dotsc,k\Delta\}$, that is
\[ N_k(\bs{x}) = \sum_{i=(k-1)\Delta+1}^{k\Delta} x_i \]
where we use the convention that $x_i=0$ for $i > m$.
The sequence of spike counts is $\bs{N}(\bs{x})=\bigl(N_1(\bs{x}),\dotsc,N_r(\bs{x})\bigr)$.  The test statistic is any function $T$ that assigns a number to each spike train.
If $\bs{X}=(X_1,\dotsc,X_m)$ is a random spike train, then the uniform interval jitter null hypothesis is
\begin{itemize}
\item[$H_0$] The distribution of $\bs{X}$ satisfies
\[ \Prob\bigl(\bs{X}=\bs{x}\bigl|\bs{N}(\bs{X})=\bs{n})\bigr) = \frac{1}{Z(\bs{n})}\ind\big\{\bs{N}(\bs{x})=\bs{n}\bigr\} \]
for all $\bs{x}$ and $\bs{n}$, where $Z(\bs{n})=\sum_{\bs{x}}\ind\{\bs{N}(\bs{x})=\bs{n}\}$ is a normalization constant.  In other words, given $\bs{N}(\bs{X})$ the conditional distribution of $\bs{X}$ is uniform. 
\end{itemize}
Each jitter window width, $\Delta$, gives a new null hypothesis (via a new definition of the sequence of spike counts, $\bs{N}$);
larger values of $\Delta$ impose more severe restrictions on temporal resolution. 

$H_0$ is a well-defined null hypothesis: every possible distribution on $m$-length binary vectors either satisfies $H_0$ or it does not.  Testing this null hypothesis is elementary: the uniform interval jitter procedure gives exact p-values for any choice of test statistic.  There are essentially no remaining statistical or mathematical decisions.  There are, however, crucial scientific decisions.
\begin{itemize}
\item What is the choice of $\Delta$?  Or is $\Delta$ a parameter to be explored?
\item What is an appropriate $T$? 
\item Are $H_0$ and $T$ an appropriate abstraction of the scientific question?
\end{itemize}
The scientist must be able to answer these questions in order to use jitter appropriately.  Relating these questions to the typical neuroscientific uses of jitter is the focus of the main text.  The focus of the Mathematical Appendix is to elaborate on some of the mathematical details and generalizations.

\section{Resampling procedures and statistical hypothesis testing} \label{s:intro}

The Appendix is organized as follows. In Section \ref{Preliminaries}, we review standard definitions and properties of hypothesis testing and of exchangeable random variables. In Section \ref{Interval_Jitter_Section}, we provide a formal description of the resampling algorithm for uniform interval jitter and relate the algorithm to an exact hypothesis test. A key mathematical justification is that, under the null hypothesis, the original and resampled data sets are exchangeable random variables. In Section \ref{Basic_Jitter_Section}, we re-present {\it basic jitter} and point out that in the case of basic jitter the exchangeability condition is violated. For this reason, we advocate the view of basic jitter as an approximate, or heuristic, variation on interval jitter. In Section \ref{Pattern_Jitter_Section}, we present pattern jitter as a generalization of uniform interval jitter which incorporates constraints on interspike intervals, and state the corresponding null hypothesis. The justification of the resulting hypothesis test is similar to that of interval jitter. In Section \ref{Tilted_Jitter}, we move beyond uniform re-sampling. Here we describe the null hypothesis for tilted jitter, whose chief characteristic is that the conditional distribution of a spike's position within a window is {\it non-uniform} (in contrast to uniform interval jitter). Here we develop variations of tilted jitter in which inference is developed for the special case of analyzing spike synchrony. Some additional mathematical details associated with tilted jitter are contained in Section \ref{s:max}. Finally, Section \ref{deterministic} develops deterministic (non-Monte Carlo) variations on some of the tests previously described.

\subsection{Preliminaries on hypothesis testing and exchangeability} \label{Preliminaries}

\subsubsection{Statistical Hypothesis Testing}

Given a data set ${\bs X}$ modelled as a random variable, a null hypothesis $H_0$ is a class of probability distributions on ${\bs X}$. A hypothesis test is a test of the hypothesis that the distribution governing ${\bs X}$ is contained in $H_0$. Formally, we will express hypothesis tests in terms of $p$-values with respect to an explicit null hypothesis $H_0$.  A statistic $p({\bs X})$ is a {\it p-value} for a null hypothesis $H_0$ if it has the property that 

\begin{equation} \label{p-value}
\Prob( p({\bs X}) \leq u ) \leq u, \forall \, 0 \leq u \leq 1, 
\end{equation}
under the condition that the distribution $\Prob$ is in $H_0$. By custom, $p$-values are reported when hypothesis tests are used. They implicitly determine critical regions: define 
\begin{equation}
C(\alpha)=\{ {\bs X}: p({\bs X}) \leq \alpha \}.
\end{equation}
By definition, if $p$ is a $p$-value (satisfying (\ref{p-value})) for $H_0$, then, under $H_0$, 
\begin{equation} \label{crit_region}
\Prob( {\bs X} \in C(\alpha) ) \leq \alpha,  
\end{equation}
so that $C(\alpha)$ is a {\it critical region} for an $\alpha$-level test of $H_0$. This is the standard definition of a test. We note that, in what follows, $p$ will sometimes depend not only on the observed data $\bs{X}$, but also on Monte Carlo resamples, say, $\bs{X}^{(1)},\dotsc,\bs{X}^{(J)}$.  In this case, we require that $\Prob\bigl( p(\bs{X},\bs{X}^{(1)},\dotsc,\bs{X}^{(J)})\leq u\bigr)\leq u$ for all $u\in[0,1]$ whenever $H_0$ is true, where the probability is with respect to the joint distribution on $\bs{X}$ and the Monte Carlo resamples.  The critical region is similarly modified.  When the number of Monte Carlo samples, $J$, is large, the additional randomness introduced by using a Monte Carlo procedure is negligible.  We note that the Monte Carlo randomness only affects power: rejecting $H_0$ when $p\leq\alpha$ correctly bounds the type I error probability below $\alpha$ for any choice of $J$.     

The set-complement of the critical region is the {\it acceptance region}. Acceptance regions for single-tests are in correspondence to what we refer to as {\it pointwise acceptance bands} when multiple tests arise from an indexed family of tests and are plotted graphically. Simultaneous acceptance bands are developed in 2.1 of this Supplement, and discussed elsewhere in the Supplement and the main text.

\subsubsection{Exchangeability}
A finite collection of random variables $Y_1, Y_2,..., Y_n$ is {\it exchangeable} if its joint distribution is invariant to permutations of its arguments. That is, 
\begin{equation}
\Prob\bigl( (Y_1,\dotsc,Y_n) \in A \bigr) = \Prob\bigl( (Y_{\pi(1)},\dotsc,Y_{\pi(n)}) \in A \bigr)
\end{equation}
for all sets\footnote{To avoid cluttering the exposition, we leave any measure-theoretic qualifications as implicit. In any case, nothing unusual comes up. Where spike times are modelled in a continuum, intervals and spike counts in intervals are, naturally, taken as finite. Wherever they appear, densities and conditional densities $f$ are for spike times, and are used in the sense of elementary probability.} $A$ and all permutations $\pi$ of the index set $(1,...,n)$.\footnote{A permutation of an ordered set is a rearrangement of the set. For example, $\pi=(4,3,1,2)$ is a permutation of $(1,2,3,4).$ Here we are writing $\pi(1)=4, \pi(2)=3,$ etc.} 

A recurring idea in the appendix is that, under a suitably-formulated null hypothesis, the collection of test statistics derived from the original spike data sets and the resampled surrogates are exchangeable. Therefore, any test of the exchangeability of the collection of test statistics will be a test of that null hypothesis. The basic hypothesis test is the following test of exchangeability. (Below we use the symbol $\ind$ to denote the indicator function.\footnote{$\ind\{A\}$ denotes the indicator function of the set $A$, which takes the value 1 when the sample outcome falls in $A$, and 0 otherwise. In particular, $\sum_{i=1}^n \ind \{Y_i \geq Y_1\}=\#\{i : Y_i \geq Y_1, 1 \leq i \leq n \}$})

\begin{lemma} \label{ExchTest}
Suppose that $Y_1,..,Y_n$ are real-valued random variables. The statistic
\begin{equation}
p(Y_1,...,Y_n) := \frac{  \sum_{i=1}^n \ind \{Y_i \geq Y_1\} }{n} 
\end{equation}
is a $p$-value for the null hypothesis that $Y_1,...,Y_n$ are exchangeable.\footnote{For completeness, here is a proof of this well known fact. Exchangeability implies that 
$\Prob(p\leq\alpha) = \Prob\bigl(\textstyle\sum_{i=1}^n \ind\{Y_i\geq Y_k\} \leq n\alpha\bigr)$
for each $k$ and therefore
$\Prob(p\leq\alpha) = \textstyle n^{-1}\sum_{k=1}^n \Prob\bigl(\textstyle\sum_{i=1}^n \ind\{Y_i\geq Y_k\} \leq n\alpha\bigr) = n^{-1}\Ex\bigl[\textstyle\sum_{k=1}^n \ind\bigl\{\textstyle\sum_{i=1}^n \ind\{Y_i\geq Y_k\} \leq n\alpha\bigr\}\bigr] \leq n^{-1}\Ex[n\alpha] = \alpha$, where $\Ex$ denotes expectation.
(The steps are easiest to verify using the order statistics, which are the result of a permutation $o(1),o(2),\ldots,o(n)$ such that 
$Y_{o(1)}\leq Y_{o(2)}\leq \cdots \leq Y_{o(n)}$.)
}
\end{lemma}

\noindent Independent and identically distributed (i.i.d.) random variables are exchangeable. We will also make use of the fact that conditionally independent and identically distributed random variables are exchangeable.
\begin{lemma} \label{CIID_Exch}
If $Y_1,...,Y_n$ are conditionally i.i.d.~(or, more generally, conditionally exchangeable), given a random variable $Z$, then $Y_1,...,Y_n$ are exchangeable.\footnote{Again, here is a proof of this well known fact.  $\Prob\bigl( (Y_1,\dotsc,Y_n) \in A \bigr) = \Ex\bigl[\Prob\bigl( (Y_1,\dotsc,Y_n) \in A \bigl|Z\bigr)\bigr] = \Ex\bigl[\Prob\bigl((Y_{\pi(1)},\dotsc,Y_{\pi(n)}) \in A \bigl| Z\bigr)\bigr] = \Prob\bigl( (Y_{\pi(1)},\dotsc,Y_{\pi(n)}) \in A \bigr)$, where the middle equality follows from the assumption that, given $Z$, the $Y$'s are exchangeable.
}
\end{lemma}

\subsection{Uniform Interval Jitter} \label{Interval_Jitter_Section}

We begin by modelling a single discretized spike train as an $m$-length Bernoulli process:  a discrete time zero-one process with outcomes in $\{0,1\}^m.$ A sample outcome is of the form ${\bs x} := (x_1,...,x_m) \in \{0,1\}^m$, where $x_j$ indicates spike or no spike in time bin $j$. In interval jitter, the time axis is partitioned into $k$ equally-sized intervals consisting of $\Delta$ bins each.\footnote{In neurophysiological experiments described in the accompanying article, spike trains were discretized into bins of 1/30 ms. When the observation interval is not divisible by $\Delta$, there will be a final interval of size less than $\Delta$. To keep things simple here, we will treat any spikes in that interval as frozen. Formally, we are additionally conditioning on the location of spikes in the final interval. But for the presentation here we will ignore this detail.}  $\Gamma_i:= \{(i-1)\Delta+1,...,i\Delta\}$ denotes the set of time bins of the $i$'th window. 

Define the sequence of spike counts $N({\bs x})=(N_1({\bs x}), N_2({\bs x}), ..., N_l({\bs x}))$ as the {\it interval-counts}:
\begin{equation}
N_i({\bs x}) := \sum_{j \in \Gamma_i} x_j
\end{equation}

Finally, let $T({\bs x}),$ the {\it test statistic}, be a fixed function that assigns a real number to a spike train (e.g. the number of synchronies relative to a second, ``reference,'' neuron).

Let ${\bs X:=(X_1,...,X_m)}$ be a random spike train drawn from an unknown probability distribution, determined by the experimental setup. A surrogate data set ${\bs X^{(i)}:=(X^{(i)}_1, X^{(i)}_2,...,X^{(i)}_m)}$ is produced by sampling from the (conditional) distribution:
\begin{equation} \label{conddist}
P( \bs{X}^{(i)}=\bs{x} | N( \bs{X} )={\bs n} )= \frac{ \ind \{ N( \bs{x} )= {\bs n} \} }{ \sum_{{\bs y} \in \{0,1\}^m} \ind \{ N({\bs y})={\bs n} \} }
\end{equation}
That is, $\bs X^{(i)}$ is drawn independently and uniformly from the subset of spike train outcomes (i.e., of $\{0,1\}^m$) that satisfy $N( {\bs X}^{(i)} )=N( {\bs X}).$ Iterating this procedure $J$ times, we obtain $J$ (conditionally) independent samples ${\bs X^{(1)}, \bs X^{(2)},...,\bs X^{(J)}}$, drawn from the distribution (\ref{conddist}). The computational procedure for sampling from (\ref{conddist}) is straightforward: proceeding independently window-by-window, assign $N_i({\bs X})$ spikes to window $i$ by sampling from the uniform distribution on the $\binom{\Delta}{N_i({\bs X})}$ possible assignments.

Finally, we define
\begin{equation} 
 p( \bs{X}, \bs X^{(1)}, \bs X^{(2)},..., \bs X^{(J)} ) = \frac{ 1 + \sum_{i=1}^J \ind \{ T( \bs X^{(i)} ) \geq T( \bs X ) \} }{J+1}.
\end{equation}

\begin{proposition} {\bf (Uniform Interval Jitter) }

\noindent Define $H_0$ as

$H_0$ : the conditional distribution of  ${\bs X}$ given  ${N(\bs X)}$ is uniform, meaning that for all ${\bs x}$ and ${\bs n}$:
\begin{equation}
\Prob( {\bs X}={\bs x} | { N(\bs X)}={\bs n} ) = \frac{ \ind \{ N(\bs{x})=\bs{n} \} } { \sum_{\bs{y} \in \{0,1\}^m} \ind \{ N(\bs{y})=\bs{n} \} }.
\end{equation}
Then, $p( \bs{X}, \bs X^{(1)}, \bs X^{(2)},..., \bs X^{(J)} )$ is a $p$-value for $H_0$, for all choices of the statistic $T$.
\end{proposition} \label{interval}
\begin{proof}
By construction and under $H_0, \bs X, \bs X^{(1)},..., \bs X^{(J)}$ are conditionally i.i.d., given $N( \bs X )$. It follows that $T({\bs X}), T(\bs X^{(1)}),..., T(\bs X^{(J)})$ are conditionally i.i.d., given $N(\bs X),$ and regardless of the choice of test statistic $T$. Therefore, the proposition follows from Lemmas \ref{ExchTest} and \ref{CIID_Exch}.
\end{proof}

The scientific motivation for the null hypothesis is treated in the main article. The following examples help to clarify its scope. 
\begin{example} 
{\bf (The case $\Delta$=1.)}
In the case $\Delta$=1, $H_0$ contains every spiking process. Thus, if we view the full hypothesis space for a 1 ms discretization as a nested hierarchy -- for example,
\begin{equation} 
\{ \Delta=100 \text{ ms}\} \subset \{\Delta=10 \text{ ms}\} \subset \{\Delta=1 \text{ ms}\},
\end{equation} 
-- then the complete hierarchy contains every probability distribution on $\bs X$. This is one of the senses in which the method is nonparametric.
\end{example}
\begin{example}
{\bf (Inhomogeneous Bernoulli Processes.)} $X_1,...,X_m$ is a Bernoulli process with parameters $p(1),...,p(m)$ where $\Pr(X_i=1)=p(i)$. The Bernoulli process is the discrete analogue of an inhomogeneous Poisson process, with $p(t)$ playing the role of the rate function. Inhomogeneous Bernoulli processes are included in the null hypothesis when the rate function $p(t)$ is constant along the intervals; i.e., when it takes the form
\begin{equation} \label{piecewise}
p(t) = \sum_{j=1}^l c_j \ind\{t \in \Gamma_j \}.  
\end{equation}
The rate function $p(t)$ can itself be random. Thus the null hypothesis includes random inhomogeneous Bernoulli processes which are drawn randomly from a distribution of rate functions satisfying (\ref{piecewise}). (The latter class corresponds to the case that the $c_j$'s are real-valued random variables with any joint distribution).
\end{example}
\begin{example}
{\bf (Deterministic $N(\bf X)$.)}
$H_0$ contains the case in which the experimental conditions exactly determine the counts in the $\Delta$ intervals, but the spike times are otherwise uniform. Thus in this case, if one took $H_0$ as an estimation model, the maximum likelihood estimate for $\Pr(N(\bf X))$ is the point distribution which produces the observed sequence of interval-counts with probability one. 
\end{example}
We note one final alternative characterization of the null. Loosely speaking, one way to motivate the perspective here is to interpret $N( \bs X )$ as containing the information about spiking structure at time scale $\Delta.$ Then, the null hypothesis is that all of the probabilistic structure in the spike train is contained in $N(\bs X)$. One formalization of that idea is through the following proposition.

\begin{proposition} {\bf (Probabilistic-invariance)}

\noindent The following conditions are equivalent:

\noindent i) $\Prob$ belongs to $H_0$. (i.e., $H_0$ is true)

\noindent ii) there exists a function $h$ such that $\Prob( \bs X=\bs x )=h( N( \bs x ) ).$

\end{proposition}
\begin{proof} If (i) is true, then (ii) follows from the definition of $H_0$ in Proposition \ref{interval}.  If (ii) is true, then clearly $\Prob(\bs{X}=\bs{x})$ is constant on $\{\bs{x}:N(\bs{x})=\bs{n}\}$ for each $\bs{n}$.  So $\Prob( \bs X=\bs x | \bs{N}(\bs{X})=\bs{n} )$ is uniform and (i) holds.
\end{proof}

\subsubsection{Multiple Spike Trains}

These methods generalize to any finite collection of $q$ spike trains $\bs X=\{\bs X_1, \bs X_2,...,\bs X_q\}$, under the same setup as the previous section.  The conditioning random variable is now $N(\bs X):=\{ N(\bs X_1); N(\bs X_2);...; N(\bs X_q)\}$, and the test statistic is any function of the collection $T( \bs X)$. Surrogate spike trains are formed as before, independently for each spike train, to form surrogate data set $\bs X^{(i)}= \{ \bs X_1^{(i)}, \bs X_2^{(i)},..., \bs X_q^{(i)}\}.$ The $p$-value is
\begin{equation}
p( \bs{X}, \bs X^{(1)}, \bs X^{(2)},..., \bs X^{(J)} ) = \frac{ 1 + \sum_{i=1}^J \ind \{ T( \bs X^{(i)} ) \geq T( \bs X ) \} }{J+1}.
\end{equation}
The statement and justification of the hypothesis test are, using the multiple spike train notation, the same as in the previous section. Generalizations to multiple spike trains of other versions of jitter (below) are equally straightforward, and will not be re-stated for each case.

\subsubsection{Continuous Time} \label{Cont_Interval}

Alternatively, one can set up jitter in the continuum. We represent the spike times $\bs t:=(t_1,t_2,...t_K)$ where $K$ is the total number of spikes in the spike train. Then let $\Omega_k:=[a_k,a_k+\Delta]$ be the {\it jitter window} associated with spike $k,$ specified by $a_k=\lfloor{t_k/\Delta}\rfloor \Delta,$ where we denote $\lfloor{x}\rfloor,$ the {\it floor} of $x,$ as the largest integer less than or equal to $x$. To form surrogate data set $\bs t^{(i)}$, we assign independently to each  $t_k^{(i)}$ a sample from the uniform distribution on $\Omega_k.$ So $t^{(i)}_1,t^{(i)}_2,...,t^{(i)}_K$, are conditionally independent, given $\Omega_1, \Omega_2,...,\Omega_K$, with conditional density
\begin{equation}
f_k(x) = \frac{ \ind \{ x \in \Omega_k \} }{\Delta}
\end{equation}
In the continuum representation, we will denote the interval counts as $N(\bs t)$. (Note that we are ``overloading'' the functions $N, T,$ and $p,$ by borrowing the notation from the discrete setup.) There is an invertible mapping from $\{\Omega_k\}_{k=1}^K$ to $N( \bs t )$ modulo permutations, so they determine an equivalent conditioning. Given a test statistic $T(\bs t)$, the $p$-value and the test are essentially the same as in the discrete case:
\begin{equation} \label{contint_pstat}
 p( \bs{t}, \bs t^{(1)}, \bs t^{(2)},..., \bs t^{(J)} ) = \frac{ 1 + \sum_{i=1}^J \ind \{ T( \bs t^{(i)} ) \geq T( \bs t ) \} }{J+1}.
\end{equation}

\begin{proposition} {\bf (Continuous Uniform Interval Jitter) }

\noindent Define $H_0$ as

$H_0$ : the conditional distribution of  ${\bs t}$ given  ${N(\bs t)}$ is uniform, meaning that its conditional density given $N(\bs t)$ is:
\begin{equation}
f( (t_1,...,t_k) | {\bs N(\bs{t})} ) = \frac{1}{\Delta^K} \prod_{k=1}^K \ind \{ t_k \in \Omega_k \}
\end{equation}
Then, $ p( \bs{t}, \bs t^{(1)}, \bs t^{(2)},..., \bs t^{(J)} )$ is a $p$-value for $H_0$, for all choices of the statistic $T$.
\end{proposition}
\begin{proof}
Under $H_0, \bs t, \bs t^{(1)},..., \bs t^{(J)}$ are conditionally i.i.d., given $N( \bs t )$, so the logic is the same as in the discrete case.
\end{proof}

\begin{example}
{\bf (Poisson and Cox Processes)}
Analogous to the inhomogeneous Bernoulli process in the discrete setup, the null includes Poisson processes with rate $\lambda(t)$ of the form
\begin{equation} \label{poisscox}
\lambda(t)=\sum_{j=1}^l c_j \ind\{t \in \Gamma_j\}.
\end{equation}
Included as well are Poisson processes with stochastic rate functions (``Cox processes''), in which the rate function is drawn randomly from a distribution of rate functions of the form of (\ref{poisscox}); that is, the $c_j$'s are random variables with any joint distribution. Intuitively, one can think about this as a test for the class of processes  well-approximated by Cox processes whose rates are  piecewise-constant in the interval partitions. This is a notion of a ``Cox process varying at timescale $\Delta.$'' (Note, however, that the null is far broader than the class of Cox processes because no constraints are placed on the distribution of spike counts in the $\Delta-$length intervals.)
\end{example}

\subsection{Basic Jitter} \label{Basic_Jitter_Section}

Basic jitter is similar to continuous interval jitter, except the jitter windows are centered around the spikes in the original data set. Following the development of Section \ref{Cont_Interval}, take instead $a_k=t_k-\Delta/2$. With respect to the resampling scheme,  everything else is the same. However, in this scheme, $\bs t, \bs t^{(1)},...,\bs t^{(J)}$ are not exchangeable regardless\footnote{Excluding the deterministic zero-spike distribution.} of the distribution $\Pr( \bs t ).$ For this reason, an analogue of Proposition \ref{interval} cannot hold for arbitrary $T$, and we recommend thinking of basic jitter as a heuristic procedure, and with some caution, particularly for complex choices of the statistic $T$. Interval jitter can be motivated as one possible way to set up a spike timing null hypothesis under which a specifiable resampling procedure preserves exchangeability.

\subsection{Pattern Jitter} \label{Pattern_Jitter_Section}

Some forms of fine temporal structure may be deemed uninteresting, and if the test statistic is sensitive to that form of structure, a jitter rejection then might be consequently uninteresting. The generic examples in the spike-resampling literature \cite{oram99} are the local interspike interval constraints, such as bursts and the refractory periods. These phenomena can impose hard or soft constraints on local interspike intervals. Pattern jitter extends the elementary spike jitter procedures so as to separate out the influence of such constraints. Here we follow the development of \cite{harrison-geman}.

Suppose an individual spike train $\bs t:=(t_1,t_2,...,t_K), $ is represented as an increasing sequence of spike times. $\bs t$ is broken down into the {\it patterns} which compose it. A pattern consists of a subsegment of the spike train for which all neighboring spikes $(t_i, t_{i+1})$ have interspike interval $t_{i+1}-t_i\leq R$, and which is separated from all spikes not in the pattern by a distance greater than $R$. The parameter $R$ thus encodes the length of history constraint which spike times must respect to preserve patterns. Under the constraint that the identity of all patterns thus defined is preserved, patterns are rigidly and uniformly jittered by jittering the pattern's first (earliest) spike uniformly within its jittering window. (The jitter windows can be chosen as in interval or basic jitter; the issues are the same: the two choices lead to procedures that are often close approximations of each other, but when windows are chosen as in interval jitter there is a corresponding exact hypothesis test.) We use $\Omega_k=[a_k,a_k+\Delta]$ to denote the jitter windows in what follows.

To formalize the idea we specify a (vector-valued) statistic $S(\bs t)$ that determines the component patterns.
\begin{equation}
S_{i1}(\bs t) := \begin{cases} a_k & \text{if $t_i-t_{i-1}>R$} \\ t_i-t_{i-1} & \text{otherwise} \end{cases}
\end{equation}
\begin{equation}
S_{i2}(\bs t) := \begin{cases} 1 & \text{if $t_i-t_{i-1}>R$} \\ 0 & \text{otherwise} \end{cases}.
\end{equation}
In words, given a spike train of $K$ spikes, $S_{k2}$ determines whether spike $k$ is the first spike in a pattern, or not. If $k$ is a first spike, $S_{k1}$ determines the (starting point of the) $\Delta-$interval that contains the first spike, otherwise $S_{k1}$ specifies the interspike interval between spike $k$ and the neighboring spike that precedes it, in its pattern. Thus $S(\bs t)$ determines the pattern-composition of a spike train, and the jitter window which contains the first spike in each pattern. Pattern jitter resamples are drawn uniformly from the set of spike outcomes consistent with $S(\bs t)$. That is, $\bs t^{(1)},\bs t^{(2)},...,\bs t^{(J)}$ are conditionally independent given $S(\bs t)$ and drawn from 
\begin{equation}
\Pr( \bs{t}^{(i)} | S(\bs t )) = \frac{1}{Z(S(\bs t))} \ind\{S(\bs t^{(i)})=S(\bs t)\},
\end{equation}
where $Z(S(\bs t))$ is the normalizing constant determined by the constraint that the conditional distribution sum to one. The computational problem of sampling from this distribution is non-trivial; an efficient procedure is developed in \cite{harrison-geman}.

If we choose the jitter windows as in interval jitter, then the function $p$ (e.g., Eq. \ref{contint_pstat}) is a $p-$value under the null hypothesis that the conditional distribution on $\bs t,$ given $S(\bs t)$ is uniform. The reasoning, via exchangeability, is just as in interval jitter. The procedure can be set up either in discrete or continuous time.

\subsection{Tilted Jitter} \label{Tilted_Jitter}

Fix $\Delta$, and consider now the continuous time representation of the spike train $\bs t = (t_1,t_2,...,t_K),$ with corresponding intervals $\Omega_1,\Omega_2,...,\Omega_K,$ with $\Omega_k=[a_k,a_k+\Delta]$.  In the conditional inference perspective developed above, a null hypothesis is specified through the conditional uniform distribution on the spike train given $N(\bs t).$ For example, under the continuous interval jitter null hypothesis, $t_1,t_2,...,t_K$ are conditionally independent given $N(\bs t)$, and the conditional density of $t_k$ given $N(\bs t)$ is taken as uniform:
\begin{equation}
f_k(x) = \frac{ \ind \{x \in \Omega_k\} } {|\Omega_k|}.
\end{equation}
A natural way to parametrically relax the hypothesis of uniformity is to broaden the uniform conditional uniform density $f_k(x)$ to a  class of smooth density functions whose variation is restricted by a parametric bound. Here we quantify variation for a density $f$ by the functional $\alpha(f)$:
\begin{equation}
\alpha(f) :=  \frac{ \max_{x \in [0,\Delta]} f(x)}{ \min_{x \in [0,\Delta]} f(x)} - 1,
\end{equation}
and define a parametric class of densities by
\begin{equation}
\mathcal{F}(S,\Delta,\epsilon) = \left\{ f \in S : \alpha(f) < \epsilon \right\},
\end{equation}
where $S$ is a specified class of smooth (for example, linear) densities on $[0,\Delta]$.

We will restrict ourselves to synchrony analysis in what follows. Let us denote a simultaneously recorded spike train $\bs s:= (s_1,s_2,...,s_{K_s})$. We will treat $\bs s$ as fixed. Defining
\begin{equation}
C := \left\{ x : \min_{1 \leq j \leq K_s} |s_j -x| \leq \delta \right\},
\end{equation}
our synchrony statistic is
\begin{equation}
v(\bs t) := \sum_{k=1}^K \ind \{t_k \in C\},
\end{equation}
which is the number of spikes in $\bs t$ that are within $\pm \delta$ of some spike in $\bs s$.

Now we will specify a resampling procedure. For each $k$, we compute
\begin{equation}
f_k^* \in \argmax_{f \in \mathcal{F}(S,\Delta,\epsilon)} \int_{C \cap \Omega_k} f(x-a_k)dx.
\end{equation}
(In the case that the $\argmax$ contains multiple elements it does not matter which density we choose). Methods of computing $f_k^*$, for various classes of smooth densities $S$, are provided in Section \ref{s:max}. To form surrogate data set $\bs t^{(i)},$ we then assign independently to each $t_k^{(i)}$ a sample from the density $f_k^*.$ The function $p$ is defined as
\begin{equation}
 p( \bs{t}, \bs t^{(1)}, \bs t^{(2)},..., \bs t^{(J)} ) = \frac{ 1 + \sum_{i=1}^J \ind \{ v( \bs t^{(i)} ) \geq v( \bs t ) \} }{J+1}.
\end{equation}

\begin{proposition} {\bf (Tilted Jitter Synchrony Test) }

\noindent Define $H_0$ as

$H_0$ : under the conditional distribution of  ${\bs t}$ given  ${N(\bs t)}$, $t_1,...,t_K$ are conditionally independent with conditional densities $f_k$ satisfying:
\begin{equation}
f_k( x_k-a_k | N(\bs t) ) \in \mathcal{F}(S,\Delta,\epsilon) \quad \forall \, k \in \{1,2,...K\}
\end{equation}
Then, $ p( \bs{t}, \bs t^{(1)}, \bs t^{(2)},..., \bs t^{(J)} ) $ is a $p$-value for $H_0$.
\end{proposition}
\begin{proof}

In this case the conditional null hypothesis is ``compound'' (containing many distributions), so exchangeability does not necessarily hold. Instead, we are drawing the surrogates from the ``worst-case'' distribution in the conditional null. 
The plan (for the proof) is to randomly perturb the spike times in each $\bs t^{(i)}$ to generate a new surrogate $\bs r^{(i)}$
in such a way that $v(\bs t), v(\bs r^{(1)}),v(\bs r^{(2)}),...,v(\bs r^{(J)})$ are conditionally i.i.d. given $N(\bs t)$,
and hence exchangeable, and furthermore 
\begin{equation} \label{e:bound}
v( \bs r^{(i)} ) \leq v( \bs t^{(i)} ), \quad \forall \, i \in \{1,2,...,J\}.
\end{equation}
Since by exchangeability $p( \bs t, \bs r^{(1)}, \bs r^{(2)},..., \bs r^{(J)} )$ is a $p$-value (under $H_0$), and since by equation \eqref{e:bound}
\begin{equation}
p( \bs t, \bs r^{(1)}, \bs r^{(2)},..., \bs r^{(J)} ) \leq p( \bs{t}, \bs t^{(1)}, \bs t^{(2)},..., \bs t^{(J)} )
\end{equation}
the proposition follows.

It remains to show how to generate $\bs r^{(i)}$ from $\bs t^{(i)}$, for each $i=1,2,\ldots,J$.
Fix a distribution $\Pr \in H_0$. Then by assumption (i.e. by $H_0$), $\{ \ind\{t_k \in C\}\}_{k=1}^K$ are conditionally independent Bernoulli random variables (r.v.'s) given $N(\bs t)$, with parameters $p_k := \Pr ( t_k=C | N(\bs t) )$. Similarly, $\{ \ind\{t_k^{(i)} \in C \}\}_{k=1}^K$ are conditionally independent Bernoulli r.v.'s, given $N(\bs t)$, with respective parameters $p_k^*:=\Pr( t_k^{(i)} \in C | N(\bs t) ).$ 
By construction, $p_k \leq p_k^*$. 
Now, for every $i=1,2,\ldots,J$ and every $k=1,2,\ldots,K$, independently
generate $r^{(i)}_k$ by the (random) prescription
\begin{equation}
r_k^{(i)} := \begin{cases} t_k^{(i)} & \text{if $t_k^{(i)} \notin C$} \\ t_k^{(i)} & \text{with probability $p_k/p_k^*$ if $t_k^{(i)} \in C$} \\ y_k & \text{otherwise}\end{cases},
\end{equation}
where $y_k$ is an arbitrarily chosen point in $\Omega_k \cap C^C$ (the distribution does not matter).
By construction, $\forall i=1,2,\ldots,J$, $\{ \ind\{r_k^{(i)} \in C\}\}_{k=1}^K$ are conditionally independent Bernoulli r.v.'s, given $N(\bs t),$ with parameters $p_1, p_2, ..., p_K$, and 
$\bs r^{(1)},\bs r^{(2)},...,\bs r^{(J)}$ are conditionally i.i.d.~given $N(\bs t)$.
Hence $v(\bs t), v(\bs r^{(1)}),v(\bs r^{(2)}),...,v(\bs r^{(J)})$
are conditionally i.i.d given $N(\bs t).$ Furthermore, $v( \bs r^{(i)} ) \leq v( \bs t^{(i)} ), \forall \, i \in \{1,2,...,J\}$ by definition.


\end{proof}



\subsection{Deterministic tests} \label{deterministic}

In several cases deterministic (non-stochastic) versions of the tests developed above are available. To give the idea, we will develop a single example here for synchrony-count test statistics, applied to a continuous uniform interval jitter test for a single spike train $\bs t.$ (A deterministic test can be developed for the tilted jitter null hypothesis in an analogous way.) The setup is as in tilted jitter (Section \ref{Tilted_Jitter}), in which we condition on the spike positions $\bs s$ of a simultaneously-observed spike train, and define
\begin{equation}
C := \left\{ x : \min_{1 \leq j \leq K_s} |s_j -x| \leq \delta \right\}
\end{equation}
through which our synchrony statistic is
\begin{equation}
v(\bs t) := \sum_{k=1}^K \ind \{t_k \in C\},
\end{equation}
which is the number of spikes in $\bs t$ that are within $\pm \delta$ of some spike in $\bs s$. Define in addition
\begin{equation}
v_k(\bs t) := \ind \{t_k \in C\},
\end{equation}
so that $v(\bs t)=\sum_{k=1}^K v_k(\bs t).$ Under $H_0, v_1(\bs t), v_2(\bs t),..., v_K(\bs t)$ are conditionally independent Bernoulli r.v.'s, given $N(\bs t)$,  with parameters $p_k(N(\bs t)):=P( v_k(\bs t) = 1 | N(\bs t) )=|C \cap \Omega_k|$. In this case, the conditional distribution of $v(\bs t)$ can be computed exactly (e.g., by convolution). It has the form
\begin{equation}
G( c; N(\bs t) ) := \Pr(v(\bs t) \geq c | N(\bs t)) := \Pr( \textstyle\sum_{k=1}^K Y_k \geq c ),
\end{equation}
where $Y_1,Y_2,...Y_{K}$ are independent Bernoulli random variables with parameters $\{p_k(N(\bs t))\}_{j=1}^K.$  Once $G$ is computed, the p-value is
\begin{equation}
G(v(\bs t);N(\bs t))
\end{equation} 

\section{Obtaining the optimization solution $f^*$} \label{s:max}

Tilted jitter requires finding an extremal spike-relocation distribution, constrained by the particular null hypothesis
being tested (see Section A.5, and Section 3.1 of the main text as well as its elaboration in this Supplement). Here we work through the calculation in a few examples.

Fix $R\subseteq[0,1]$ and let $f$ denote a generic probability density function (pdf) on $[0,1]$.  We want to maximize $\int_R f$ subject to the constraint that $\alpha(f)\leq\epsilon$ for some fixed $\epsilon \geq 0$, where
\[ \alpha(f) := \frac{\max_{0\leq x \leq 1} f(x)}{\min_{0\leq x \leq 1} f(x)}-1  , \]
and perhaps also subject to additional smoothness constraints, like linearity of $f$.  Define $M_\epsilon := \{ f:\alpha(f)\leq\epsilon\}$ and let $S$ be the set of pdfs that satisfy the smoothness constraints.  Then the problem of interest is to find any
\[ f^* \in \argmax_{f\in M_\epsilon\cap S} \int_R f(x) dx . \]
We always assume that the uniform (constant) pdf satisfies the smoothness constraints so that $M_\epsilon\cap S\neq\emptyset$.  If $\int_R dx \in \{0,1\}$, then any $f$ is an $\argmax$ and in practice we take $f^*$ to be the uniform pdf.  Henceforth, we assume that 
\[ 0 < |R| := \int_R dx < 1 . \]

Note that $\int_R f$ is linear in $f$ and that $M_\epsilon$ is closed and convex.  If $S$ is also closed and convex, then $M_\epsilon\cap S$ is closed and convex and a maximizer $f^*$ exists and can be found on the extremal points of $M_\epsilon\cap S$. 

\subsection{No smoothness constraints}

If $S:=\{\text{all pdfs on $[0,1]$}\}$ so that we are maximizing over $M_\epsilon$, then a solution is
\[ f^*(x) := \frac{1+\epsilon\ind\{x\in R\}}{1+\epsilon|R|} , \]
where $\ind\{A\}$ is the indicator function of the event $A$ and where $x\in[0,1]$.
\begin{proof}
It is easy to verify that $\alpha(f^*)=\epsilon$ and that 
\begin{equation} \int_R f^*(x) dx = \frac{(1+\epsilon)|R|}{1+\epsilon |R|} := \beta . \label{e:1} \end{equation}
Now suppose that $f$ is a pdf with $\int_R f > \beta$.  This gives 
\[ \max_x f(x) \geq \frac{\int_R f(x) dx}{\int_R dx} > \frac{\beta}{|R|} = \frac{1+\epsilon}{1+\epsilon|R|} \]
and
\[ \min_x f(x) \leq \frac{\int_{R^c} f(x) dx}{\int_{R^c} dx} < \frac{1-\beta}{1-|R|} = \frac{1}{1+\epsilon|R|} . \]
Combining these gives $\alpha(f) > \epsilon$ and shows that $f^*$ achieves the maximum.
\end{proof}

\subsection{Finite mixture smoothness constraints}

Let $\{g_1,\dotsc,g_m\} \in M_\epsilon$ be fixed and known pdfs.\footnote{If it is not known that each $g_k\in M_\epsilon$, then define $h_k=g_k-1$ and use the methods in Section \ref{s:basis}, perhaps with the additional restriction that $a$ is a probability vector.}  Suppose that $S$ is the set of mixtures of the $g_k$'s, namely,
\[ S:=\left\{ \sum_{k=1}^m p_k g_k : \sum_{k=1}^m p_k = 1 \text{ and each } p_k \geq 0\right\} . \]
Note that $S\subset M_\epsilon$, so we are maximizing over $S$.  Since $S$ is closed and convex, we can take $f^*$ to be an extremal point.  The extremal points of $S$ are simply $\{g_1,\dotsc,g_m\}$.  To find $f^*$ we compute $\int_R g_k$ for each $k$ and choose the one that gives the largest integral.
 
\subsection{Finite basis smoothness constraints} \label{s:basis}

Let $\{h_1,\dotsc,h_m\}$ be fixed and known functions on $[0,1]$ that each integrate to $0$.  For any $a:=(a_1,\dotsc,a_m)\in\mathbb{R}^m$ define the function $f_a$ on $[0,1]$ by
\[ f_a(x) := 1 + \sum_{k=1}^m a_k h_k(x)  \]
and let
\[ A := \left\{a\in\mathbb{R}^m : \min_{0\leq x \leq 1} f(x) \geq 0\right\} , \]
so that $f_a$ is a valid pdf whenever $a\in A$.  Finally, define
\[ A_\epsilon := \left\{a\in A : \alpha(f_a) \leq \epsilon\right\} . \]
  
Suppose that $S:=\{f_a : a\in A\}$.  Notice that optimizing over $M_\epsilon\cap S$ is equivalent to optimizing over $A_\epsilon$.  The computational problem, then, is to find a maximizer of the function
\[ \int_R f_a(x) dx = \int_R dx + \sum_{k=1}^m a_k \int_R h_k(x) dx \]
over $a\in A_\epsilon$.  Defining $H:=(H_1,\dotsc,H_m)$ for $H_k := \int_R h_k(x) dx$, which are known, the computational problem reduces to finding
\[ a^* \in \argmax_{a\in A_\epsilon} \ a\!\cdot\!H , \]
where $a\!\cdot\!H$ denotes the dot product.
Note that $a\!\cdot\!H$ is linear in $a$ and $A_\epsilon$ is closed and convex, so we can restrict attention to the extremal points of $A_\epsilon$.

There exist generic algorithms for solving convex optimization problems of this kind, although for specific instances it can be computationally much quicker to use custom algorithms based on further analysis.  

\subsubsection{Linear smoothness constraints}

A simple example is the case where $m=1$ and $h_1(x):=2x-1$, so that $S$ is the class of linear pdfs.  Noting that  
\[ \alpha(f_a) = \frac{1+|a|}{1-|a|}-1 \]
we see that 
\[ A_\epsilon = \left[-\frac{\epsilon}{\epsilon+2},\frac{\epsilon}{\epsilon+2}\right]  , \]
so the extremal points of $A_\epsilon$ are the two points $-\epsilon/(\epsilon+2)$ and $\epsilon/(\epsilon+2)$.  The values of $a\!\cdot\!H$ at these extremal points are $-H\epsilon/(\epsilon+2)$ and $H\epsilon/(\epsilon+2)$, respectively.  Choosing the maximum and then the corresponding $f_a$ gives the solution
\[ f^*(x) = \begin{cases} 1+\displaystyle\frac{\epsilon}{\epsilon+2\vphantom{\bigl)}}(2x-1) & \text{if $\int_R (2x-1) dx > 0$,} \\ 
1-\displaystyle\frac{\epsilon\vphantom{\bigl)}}{\epsilon+2\vphantom{\bigl)}}(2x-1) & \text{if $\int_R (2x-1) dx < 0$,} \\
1\vphantom{\bigl)} & \text{otherwise .} \end{cases} \]

\end{document}